\documentclass[11pt]{article}

\usepackage[utf8]{inputenc}
\usepackage[T1]{fontenc}
\usepackage{amsmath,amssymb,amsthm,mathtools}
\usepackage{mathrsfs}
\usepackage{booktabs}
\usepackage{enumitem}
\usepackage{geometry}
\usepackage{microtype}
\usepackage[unicode,hidelinks]{hyperref}
\usepackage{bookmark}

\numberwithin{equation}{section}
\allowdisplaybreaks
\newtheorem{theorem}{Theorem}[section]
\newtheorem{lemma}[theorem]{Lemma}
\newtheorem{proposition}[theorem]{Proposition}
\newtheorem{corollary}[theorem]{Corollary}
\theoremstyle{definition}

\newtheorem{example}[theorem]{Example}
\theoremstyle{remark}
\newtheorem{remark}[theorem]{Remark}

\newcommand{\F}{\mathbb F}
\newcommand{\Z}{\mathbb Z}
\newcommand{\A}{\mathbb A}
\newcommand{\PP}{\mathbb P}
\newcommand{\Gm}{\mathbb G_m}
\newcommand{\Gr}{\operatorname{Gr}}
\newcommand{\Spec}{\operatorname{Spec}}
\newcommand{\Tr}{\operatorname{Tr}}
\newcommand{\wt}{\operatorname{wt}}
\newcommand{\rwt}{\operatorname{rwt}}
\newcommand{\cont}{\operatorname{cont}}
\newcommand{\ord}{\operatorname{ord}}
\newcommand{\coeff}{\operatorname{coeff}}
\newcommand{\Duniv}{\mathscr D_{n,r}}
\newcommand{\Ddet}{\mathscr D^{\mathrm{det}}_{n,r}}
\newcommand{\DSch}{\mathscr D^{\mathrm{Sch}}_{n,r}}
\newcommand{\Schur}[1]{\mathscr S_{#1}}
\newcommand{\Part}{\mathcal P_{n,r}}
\newcommand{\MDSlocus}{\mathcal M_{n,r}}
\newcommand{\GRSlocus}{\mathcal G_{n,r}}

\hypersetup{
  pdftitle={Schur–Plücker Geometry of the MDS Locus for Principal-Ideal Codes},
  pdfauthor={Yangcheng Li and Pingzhi Yuan},
  pdfkeywords={MDS codes; principal-ideal codes; Schur polynomials; Plücker coordinates; Grassmannians; confluent Vandermonde determinants; sparse multiples},
  pdfdisplaydoctitle=true,
  bookmarksnumbered=true,
  bookmarksopen=true
}

\title{Schur--Pl\"ucker Geometry of the MDS Locus\\
for Principal-Ideal Codes}
\author{Yangcheng Li\thanks{School of Mathematical Sciences, South China Normal University, Guangzhou 510631, China. Email: \texttt{liyc@m.scnu.edu.cn}.}
\and
Pingzhi Yuan\thanks{School of Mathematical Sciences, South China Normal University, Guangzhou 510631, China. Email: \texttt{yuanpz@scnu.edu.cn}.}}
\date{}

\begin{document}
\maketitle

\begin{abstract}
Let \(g\) be a monic polynomial of degree \(r<n\), and let \(C_g(n)\) be the coefficient-vector code formed by multiples \(ug\) with \(\deg(ug)<n\). We study the coefficient-space MDS locus \(M_{n,r}\). The companion construction identifies coefficient space with the moduli of cyclic matrix–vector pairs, and the remainder-orbit map embeds it as a smooth complete intersection in the standard big cell of \(\operatorname{Gr}(r,n)\).

We prove that every normalized maximal Plücker coordinate pulls back, up to sign, to a power of the constant coefficient \(A_0\) times a Schur polynomial \(S_\kappa(g)=s_\kappa(\Lambda_g)\), where \(\kappa\subseteq (n-r)^{r-1}\). Hence the universal MDS polynomial is
\[
D_{n,r}=A_0\prod_{\kappa\subseteq (n-r)^{r-1}}S_\kappa.
\]
This description yields a flat non-MDS boundary over \(\mathbb Z\), explicit degree and finite-field estimates, and a length filtration governed by sparse multiples. It also gives bad-characteristic criteria on root-multiplicity strata and density-one results on the irreducible stratum. 
Finally, for $r\ge 3$ and $N\ge r+3$, every nonempty first-failure layer over an algebraically closed field has a dense open non-GRS locus.
\end{abstract}

\noindent\textbf{Keywords:} MDS codes; principal-ideal codes; Schur
polynomials; Pl\"ucker coordinates; \mbox{Grassmannians}; confluent Vandermonde
determinants; sparse multiples.

\noindent\textbf{2020 Mathematics Subject Classification:} 94B05; 05E05;
11T71; 14M15.

\section{Introduction}
\label{sec:introduction}

Let $K$ be a field and let
\begin{equation}
 g(x)=x^r+a_{r-1}x^{r-1}+\cdots+a_1x+a_0\in K[x],
 \qquad 1\le r<n.
 \label{eq:intro-g}
\end{equation}
The associated coefficient-vector code is
\begin{equation}
 \mathcal C_g(n)=
 \{u(x)g(x):u(x)\in K[x],\ \deg u<n-r\}\subseteq K^n,
 \label{eq:intro-code}
\end{equation}
where a polynomial of degree less than $n$ is identified with its coefficient
 vector.  It has dimension $n-r$.  To make the principal-ideal terminology
 explicit, put $s=n-r$, choose any monic $v\in K[x]$ of degree $s$, and set
 $h=gv$.  Multiplication by $g$ induces an isomorphism of $K[x]$-modules
 (and hence of $K$-vector spaces)
\begin{equation}
 K[x]/(v)\xrightarrow{\ \sim\ }(\bar g)\subseteq K[x]/(h),
 \qquad \bar u\longmapsto\overline{ug},
 \label{eq:intro-principal-ideal-realization}
\end{equation}
because $ug\equiv0\pmod{gv}$ if and only if $v\mid u$.  The unique
representatives with $\deg u<s$ map to the degree-$<n$ coefficient vectors
in~\eqref{eq:intro-code}.  Thus the ambient quotient realization may depend
on $v$, but the code and all remainder-coordinate tests below depend only on
$(g,n)$.  We therefore begin with the companion remainder orbit rather than
with the quotient ring.

Codes realized as ideals in polynomial factor rings are commonly studied as
polycyclic codes; see
\cite{LopezPermouthEtAl,LopezPermouthParraSzabo}.  The present family is
distinguished by the coefficient-vector realization above, which depends
only on $(g,n)$ although the ambient factor ring may vary with $v$.

Throughout, over an arbitrary field $K$, an $[n,n-r]$ code is said to be
\emph{of GRS type} if it is monomially equivalent to a homogeneous
(projective) extended generalized Reed--Solomon code evaluated at pairwise
distinct $K$-rational points of $\PP^1_K$.  For $r\ge2$, equivalently, the
projective columns of a full-rank parity-check matrix are pairwise distinct
and lie on a split rational normal curve in $\PP_K^{r-1}$.  Li and
Yuan~\cite[Proposition~2.5]{LiYuanCyclicOrbits} prove the corresponding
projective characterization over finite fields, and their argument applies
verbatim over an arbitrary field $K$.  A code is called \emph{non-GRS} if it
is not of GRS type in this sense.  Over finite fields this agrees with the
convention in~\cite{LiYuanCyclicOrbits}.
The polynomial-evaluation construction goes back to Reed and
Solomon~\cite{ReedSolomon}.  For the classical correspondence between linear
MDS codes and projective arcs, and the role of rational normal curves, see
\cite{AldersonBruenSilverman,BallLavrauw,Segre}.

Li and Yuan~\cite[Proposition~3.1 and Theorem~3.2]{LiYuanCyclicOrbits}
construct a canonical remainder-coordinate parity-check matrix
$H_g^{\mathrm{rem}}$ and introduce an exact MDS polynomial $D_{n,r}(g)$
satisfying
\begin{equation}
 \mathcal C_g(n)\text{ is MDS}
 \quad\Longleftrightarrow\quad D_{n,r}(g)\ne0.
 \label{eq:intro-LiYuan-MDS}
\end{equation}
They also classify, for $r\ge3$ and $n\ge r+3$, the coefficient points for
which the code is of GRS type, using cyclic projective orbits and
rational-normal-curve rigidity
\cite[Theorems~5.3 and~6.6]{LiYuanCyclicOrbits}.

Here we study a different aspect of the same code family.  For each field
$K$, write
\[
 \GRSlocus(K)
 \subseteq\mathcal M_{n,r}(K)\subseteq K^r
\]
for the GRS-type coefficient set and the set of $K$-points of the
coefficient-space MDS locus of this principal-ideal family, respectively.
The notation $\GRSlocus(K)$ is purely set-theoretic; no $\Z$-subscheme is
being asserted.  The classification in~\cite{LiYuanCyclicOrbits} studies the
former set through rational-normal-curve geometry, whereas our focus is the
internal geometry of $\MDSlocus$, using Grassmannians, Pl\"ucker coordinates,
Schur functions, arithmetic factorization strata, length filtrations, and
sparse multiples.

Except where we explicitly invoke~\cite{LiYuanCyclicOrbits} for the
pure-power criterion and the GRS classification and counting results, the
Schur--Pl\"ucker, divisor, filtration, and sparse-threshold arguments below
are self-contained.

The universal determinantal product
$D_{n,r}=A_0\prod_U\delta_U$, introduced in
\cite[Equation~(3.2) and Theorem~3.2]{LiYuanCyclicOrbits}, is nonzero exactly
when $\mathcal C_g(n)$ is MDS.  We denote this same polynomial by $\Ddet$
when emphasizing its determinantal realization; the superscript does not
indicate a new lift or a new MDS polynomial.  The new result here is its
identification with the rectangular Schur realization $\DSch$.  Its behavior
on factorization strata and its variation with $n$ lead to three questions:
\begin{enumerate}[label=\textup{(Q\arabic*)},leftmargin=3.2em]
\item Do all Pl\"ucker coordinates of the remainder orbit admit a uniform
symmetric-function description?
\item How does the MDS open set meet prescribed root-multiplicity and
irreducible Frobenius strata?
\item As $n$ increases, which new equations enter, and at which length does
the MDS property fail for the first time?
\end{enumerate}

Our answer to the first question is the following master theorem.  Put
\[
 H_g^{\mathrm{rem}}=(Q_0\ Q_1\ \cdots\ Q_{n-1}),
\]
where $Q_i$ is the coefficient vector of $x^i\bmod g$.  For
$I=\{i_0<i_1<\cdots<i_{r-1}\}$ set
\[
 U_I=\{i_1-i_0,\ldots,i_{r-1}-i_0\}
\]
and, for $U=\{u_1<\cdots<u_{r-1}\}$, define
\[
 \kappa(U)=
 \bigl(u_{r-1}-(r-1),u_{r-2}-(r-2),\ldots,u_1-1,0\bigr).
\]
For $r=1$, all index lists are empty and we use
$\kappa(\varnothing)=\varnothing$ and $\Schur{\varnothing}=1$.
If $\Lambda_g$ is the root multiset of $g$, let
$\Schur{\kappa}(g)=s_\kappa(\Lambda_g)$.

\begin{theorem}[Rectangular Schur--Pl\"ucker theorem]
\label{thm:intro-master}
For every field, every monic $g$ of degree $r$, and every $r$-subset
$I\subseteq\{0,\ldots,n-1\}$,
\begin{equation}
 \det(Q_{i_0},\ldots,Q_{i_{r-1}})
 =\bigl((-1)^ra_0\bigr)^{i_0}\Schur{\kappa(U_I)}(g).
 \label{eq:intro-all-pluecker}
\end{equation}
The reduced index sets are in bijection with all partitions in the rectangle
$(n-r)^{r-1}$.  If
$\DSch:=A_0\prod_{\kappa\subseteq(n-r)^{r-1}}\Schur{\kappa}$, then
\begin{equation}
 \boxed{\ \Ddet=\DSch\ },
 \label{eq:intro-product}
\end{equation}
and we denote their common value by $\Duniv$.  Moreover,
\begin{equation}
 \mathcal C_g(n)\text{ is MDS}
 \quad\Longleftrightarrow\quad
 a_0\ne0\ \text{and}\quad
 \Schur{\kappa}(g)\ne0
 \quad\text{for every }\kappa\subseteq(n-r)^{r-1}.
 \label{eq:intro-schur-MDS}
\end{equation}
\end{theorem}

Although written in terms of a root multiset, each $\Schur{\kappa}(g)$ lies
in the coefficient field and is independent of the chosen splitting field
and of the ordering of the roots.
Thus~\eqref{eq:intro-product} identifies the determinantal and Schur
realizations of the same exact MDS polynomial.
The generalized and confluent bialternant identities used in the proof are
classical or known~\cite{FloweHarris,GonzalezSerranoMaximenko,LiLin,Macdonald}.
The new point is not the isolated identity ``generalized Vandermonde equals
Vandermonde times Schur.''  It is the identification of \emph{all} reduced
MDS minors with exactly one complete rectangle of partitions, the extension
to every Pl\"ucker coordinate of the companion remainder orbit, and the
coding-theoretic and arithmetic consequences of that complete system.
In particular, after separating the constant-coefficient factor $A_0$, the
MDS condition is governed by a single rectangular Schur arrangement rather
than by an unrelated collection of minor equations.

Before passing to the Grassmannian, coefficient space has an intrinsic
moduli interpretation.  The scheme of pairs $(A,z)$ for which
$z,Az,\ldots,A^{r-1}z$ is a basis is
$\operatorname{GL}_{r,\Z}\times\A^r_{\Z}$; simultaneous change of basis acts
by left translation on the first factor.  Thus its fpqc quotient is
$\A^r_{\Z}$, and the companion pair is a global slice.

The theorem places this coefficient family in the Grassmannian through
\[
 \Phi_{n,r}:\A^r_{\Z}\longrightarrow\Gr(r,n),
 \qquad g\longmapsto\operatorname{rowspan}H_g^{\mathrm{rem}}.
\]
The first $r$ columns form the identity.  More precisely, $\Phi_{n,r}$
factors as a closed immersion of $\A^r_{\Z}$ into the standard big cell;
its composite with the big-cell open immersion is therefore a locally
closed immersion into $\Gr(r,n)$.  In big-cell coordinates its image is cut
out by $r(n-r-1)$ triangular companion-recurrence equations forming a
regular sequence.  It is therefore an explicit smooth closed complete
intersection, and its intersection with the uniform-matroid open set is
isomorphic to $\MDSlocus$.  The canonical principal effective Cartier
boundary divisor is flat over $\Spec\Z$ and has support
\[
 V(A_0)\ \cup\!
 \bigcup_{\kappa\subseteq(n-r)^{r-1}}V(\Schur{\kappa}).
\]
This is a canonical Schur-coordinate decomposition, not a claim that every
displayed factor is irreducible in the coefficient ring.
The pullback of the product of \emph{all} Pl\"ucker coordinates records the
multiplicities not recorded when one retains one representative from each
translation class of minors:
\[
 \binom n{r+1}\operatorname{div}(A_0)
 +\sum_{\kappa\subseteq(n-r)^{r-1}}
   (n-r+1-\kappa_1)\operatorname{div}(\Schur{\kappa}).
\]
It has the same support as the canonical boundary and ordinary degree
$r\binom n{r+1}$.

The remaining results develop four consequences of this description.
\begin{enumerate}[label=\textup{(\roman*)},leftmargin=2.8em]
\item The universal boundary polynomial is primitive over $\Z$, and the
associated divisor, as well as each of its outer-layer increments, is flat
over $\Z$.  For $r\ge2$ the increments have nonzero fibers in every
characteristic; for $r=1$ they are empty.
Its exact total degree is
\[
 \delta_{n,r}=
 1+\frac{(r-1)(n-r)}{r}\binom{n-1}{r-1},
\]
and its root-weighted degree is computed exactly.  This yields the explicit bound
$\#\MDSlocus(\F_q)\ge q^r-\delta_{n,r}q^{r-1}$.  Root dilation and
reciprocity arise from explicit monomial and permutation equivalences of
codes, while the Schur identities describe their action factor by factor.
For $r\ge2$, the factor $\Schur{(1)}=-A_{r-1}$ is a unit on the MDS locus,
and normalization of $A_{r-1}$ gives a scheme isomorphism
$\MDSlocus\simeq\Gm\times\mathcal M_{n,r}^{(1)}$ with explicit quotient
coordinates.  In particular, its number of $\F_q$-points is divisible by
$q-1$.
Combining the estimate with the exact GRS count in
\cite[Theorem~6.14]{LiYuanCyclicOrbits} gives a quantitative lower bound for
MDS non-GRS members.  A separate avoidance argument gives an explicit lower
bound for completely split MDS non-GRS polynomials.

 \item The partition rectangles are nested with $n$, producing, for $r\ge2$,
 a filtration of effective boundary divisors with nonzero flat Cartier
 increments in every characteristic.  This divisor statement alone does not
 assert strict containment of supports.  The corresponding MDS open sets form
 a decreasing filtration; each first-failure subscheme is a flat effective
 Cartier divisor in the preceding MDS open, with underlying space equal to
 the corresponding set-theoretic difference.  The layer at index $N$ is
 faithfully flat over $\Z[1/N]$, so the corresponding inclusion of MDS opens
 is geometrically strict after base change to every field of characteristic
 zero or prime to $N$; when $r\ge3$ and $N\ge r+3$, we count its GRS-type
 coefficient points exactly.  After base change to an algebraically closed
 field, for $r\ge3$ and $N\ge r+3$, the closure of the GRS-type coefficient
 set in $\mathcal X_{N,r}$ has dimension at most one, hence codimension at
 least $r-2$ in every component it meets; the non-GRS locus is dense in every
 component of every nonempty such layer.

\item For every multiplicity type
$\mathbf m=(m_1,\ldots,m_s)$, the pullback of $\MDSlocus$ to the ordered
root stratum is cut out by
\[
 \prod_{\kappa\subseteq(n-r)^{r-1}}
 s_\kappa(Y_1^{[m_1]},\ldots,Y_s^{[m_s]}).
\]
This gives an exact content criterion: the prime divisors of the contents of
these factors are exactly the
characteristics in which the entire stratum collapses into the non-MDS
boundary, and every such prime is smaller than $n$.  Away from those primes
we give a quantitative density estimate.  In positive characteristic, every
non-squarefree multiplicity type has a finite structural bad-length bound,
with exact formulas for pure powers and the type $(2,1)$.  On the irreducible
stratum the master theorem becomes a
Frobenius--Schur criterion.  We prove explicit large-field estimates and
 density one for irreducible MDS members as $q\to\infty$, both for fixed
 $(n,r)$ and, with $r$ fixed, for $n=n(q)$ satisfying $n(q)^r=o(q)$.  When
 $r\ge3$ and $n\ge r+3$, these irreducible MDS members are automatically
 non-GRS by~\cite[Corollary~5.5]{LiYuanCyclicOrbits}.

\item Over $\F_q$, the order of the companion class in
$\operatorname{PGL}_r(\F_q)$ is computed exactly from root ratios and
multiplicities.  It equals the least exponent of a binomial $x^N-c$
divisible by $g$ and gives an explicit upper bound for the sparse threshold.
The first vanishing Schur width equals, up to the shift $r-1$, the degree of
the first nonzero multiple of $g$ having at most $r$ nonzero coefficients.
On every generic multiplicity stratum this sparse threshold equals the
structural bad-length threshold minus one.
\end{enumerate}

The organization follows this chain.  Section~\ref{sec:setup} constructs the
remainder and Grassmannian models and fixes Schur notation.
Section~\ref{sec:schur-system} proves the rectangular Schur--Pl\"ucker
theorem and its confluent form.  Section~\ref{sec:global-arithmetic} studies
the universal MDS boundary divisor and finite-field estimates.
Section~\ref{sec:strata} treats multiplicity and irreducible arithmetic
strata.  Section~\ref{sec:length} develops the length filtration and sparse
thresholds.  Section~\ref{sec:examples} gives specializations and examples
organized by the rectangular Schur system.  Appendix~A
contains the detailed quadratic Frobenius reconstructions used in the
coefficient tests of Section~\ref{sec:strata}.

\section{Remainder orbits and the Grassmannian model}
\label{sec:setup}

\subsection{Coefficient codes and remainder coordinates}

Let
\[
 R_r=\Z[A_0,\ldots,A_{r-1}]
\]
and consider the universal monic polynomial
\begin{equation}
 \mathbf g(x)=x^r+A_{r-1}x^{r-1}+\cdots+A_1x+A_0\in R_r[x].
 \label{eq:universal-g}
\end{equation}
For a specialization $g$ as in~\eqref{eq:intro-g}, write
\begin{equation}
 x^i\equiv b_{0,i}+b_{1,i}x+\cdots+b_{r-1,i}x^{r-1}\pmod g,
 \qquad
 Q_i=(b_{0,i},\ldots,b_{r-1,i})^T.
 \label{eq:remainder-vector}
\end{equation}
Thus $Q_0,\ldots,Q_{r-1}$ are the standard basis vectors.  If $T_g$ is the
matrix of multiplication by $x$ on $K[x]/(g)$ in the power basis, then
\begin{equation}
 Q_i=T_g^iQ_0,
 \qquad
 \det(T_g)=(-1)^ra_0,
 \label{eq:companion-shift}
\end{equation}
and the columns satisfy the recurrence
\begin{equation}
 Q_{i+r}=-a_{r-1}Q_{i+r-1}-\cdots-a_1Q_{i+1}-a_0Q_i.
 \label{eq:remainder-recurrence}
\end{equation}

For a polynomial $f(x)=\sum_i c_ix^i$, we write
\begin{equation}
 \wt(f):=\#\{i:c_i\ne0\}
 \label{eq:coefficient-weight}
\end{equation}
for its coefficient weight.

The next proposition combines
\cite[Proposition~3.1 and Theorem~3.2\textup{(i)--(iii)}]{LiYuanCyclicOrbits};
we include the short argument to fix the remainder notation used throughout.

\begin{proposition}[Remainder parity check and sparse multiples]
\label{prop:remainder-parity-sparse}
The matrix
\begin{equation}
 H_g^{\mathrm{rem}}=(Q_0\ Q_1\ \cdots\ Q_{n-1})
 \label{eq:remainder-matrix}
\end{equation}
is a parity-check matrix of $\mathcal C_g(n)$.  Consequently,
\begin{align}
 \mathcal C_g(n)\text{ is MDS}
 &\Longleftrightarrow
 \text{every $r$ columns of $H_g^{\mathrm{rem}}$ are independent},
 \label{eq:parity-MDS}\\
 &\Longleftrightarrow
 \nexists\,0\ne f\in K[x]
 \text{ with }g\mid f,\ \deg f<n,\ \wt(f)\le r.
 \label{eq:sparse-criterion}
\end{align}
\end{proposition}

\begin{proof}
For $c(x)=\sum_{i=0}^{n-1}c_ix^i$, the coordinate vector of the remainder
of $c$ modulo $g$ is $\sum_i c_iQ_i$.  Since $\deg c<n$,
\[
 c\in\mathcal C_g(n)
 \Longleftrightarrow g\mid c
 \Longleftrightarrow \sum_i c_iQ_i=0.
\]
The first $r$ columns form the identity, so the matrix has rank $r$ and is a
parity check.  The first equivalence is the standard parity-check
characterization of an $[n,n-r]$ MDS code.  The minimum distance is the
least coefficient weight of a nonzero degree-$<n$ multiple of $g$, which
gives the second equivalence.
\end{proof}

For later use, if $g$ is monic of degree $r$ with $a_0\ne0$, define its
sparse-multiple degree by
\begin{equation}
 \sigma_r(g)
 :=\min\{\deg f:0\ne f,\ g\mid f,\ \wt(f)\le r\},
 \label{eq:sigma-definition}
\end{equation}
with value $\infty$ if the displayed set is empty.

The sparse-multiple problem for general input polynomials is algorithmically
subtle~\cite{GiesbrechtRocheTilak}.  In the present companion-orbit setting,
\eqref{eq:sparse-criterion} will instead lead to an exact threshold governed
by the first vanishing rectangular Schur coordinate.

\subsection{Cyclic pairs and the companion slice}

Over a field, the matrix $[\,z\ Az\ \cdots\ A^{r-1}z\,]$ is the single-input
reachability (Krylov) matrix, and its relation to cyclic matrices and
companion normal form is classical~\cite{FerranteWimmer}.  The change-of-basis
and moduli viewpoint is part of algebraic systems
theory~\cite{HazewinkelKalman}.  The companion pair $(T_g,Q_0)$ is the
universal normal form for a matrix together with a cyclic vector.  The
following statement gives the corresponding scheme-theoretic assertion over
$\Z$ and exhibits the companion pair as a global slice.

\begin{proposition}[Cyclic pairs and the companion slice]
\label{prop:cyclic-pair-scheme}
Let $\operatorname{Mat}_r:=\A_{\Z}^{r^2}$, let
$e_0,\ldots,e_{r-1}$ be the standard column basis, and put
\[
 B(A,z):=[\,z\ Az\ \cdots\ A^{r-1}z\,].
\]
The cyclic-pair scheme
\[
 \operatorname{Cyc}_r
 :=D\!\left(\det B(A,z)\right)
 \subseteq\operatorname{Mat}_r\times\A_{\Z}^r
\]
is isomorphic to $\operatorname{GL}_{r,\Z}\times\A_{\Z}^r$.  More precisely,
for $a=(a_0,\ldots,a_{r-1})$, let $T(a)$ be the companion matrix determined by
\[
 T(a)e_i=e_{i+1}\quad(0\le i<r-1),
 \qquad
 T(a)e_{r-1}=-\sum_{i=0}^{r-1}a_i e_i.
\]
Then
\begin{equation}
 \Theta:\operatorname{GL}_{r,\Z}\times\A_{\Z}^r
 \xrightarrow{\ \sim\ }\operatorname{Cyc}_r,
 \qquad
 (P,a)\longmapsto\bigl(PT(a)P^{-1},Pe_0\bigr)
 \label{eq:cyclic-pair-trivialization}
\end{equation}
is an isomorphism, with inverse
\begin{equation}
 (A,z)\longmapsto\bigl(B(A,z),a(A,z)\bigr),
 \qquad
 \begin{pmatrix}a_0(A,z)\\ \vdots\\ a_{r-1}(A,z)\end{pmatrix}
 =-B(A,z)^{-1}A^rz.
 \label{eq:cyclic-pair-inverse}
\end{equation}
For the simultaneous change-of-basis action
\[
 U\cdot(A,z)=(UAU^{-1},Uz),
\]
the isomorphism~\eqref{eq:cyclic-pair-trivialization} identifies the action
with $U\cdot(P,a)=(UP,a)$.  Hence the projection
\[
 \pi:\operatorname{Cyc}_r\longrightarrow\A_{\Z}^r,
 \qquad
 (A,z)\longmapsto a(A,z),
\]
is a trivial $\operatorname{GL}_{r,\Z}$-torsor, and
\[
 \operatorname{Cyc}_r/\operatorname{GL}_{r,\Z}\simeq\A_{\Z}^r
\]
for the fpqc sheaf quotient (and for the categorical quotient).  Moreover,
the quotient coordinates are the characteristic-polynomial coefficients:
\[
 \det(xI_r-A)=x^r+a_{r-1}(A,z)x^{r-1}+\cdots+a_1(A,z)x+a_0(A,z).
\]
All assertions are preserved by arbitrary base change.
\end{proposition}

\begin{proof}
Let $R$ be a commutative ring and $(A,z)\in\operatorname{Cyc}_r(R)$.  Since
$B:=B(A,z)$ is invertible, there is a unique $a\in R^r$ such that
\[
 A^rz=-\sum_{i=0}^{r-1}a_iA^iz;
\]
its coordinate column is exactly~\eqref{eq:cyclic-pair-inverse}.  Comparing
columns gives $AB=BT(a)$, and therefore
\[
 A=BT(a)B^{-1},\qquad z=Be_0.
\]
Conversely, if $(A,z)=(PT(a)P^{-1},Pe_0)$, then $A^jz=Pe_j$ for
$0\le j<r$, so $B(A,z)=P$.  Thus the formulas are mutually inverse on
$R$-points.  They are natural in $R$, and the entries of $B(A,z)^{-1}$ are
regular on $D(\det B(A,z))$; hence they define mutually inverse morphisms of
$\Z$-schemes.

The companion convention gives
\[
 \det(xI_r-T(a))=x^r+a_{r-1}x^{r-1}+\cdots+a_0,
\]
which proves the characteristic-polynomial assertion.  Finally,
$U\cdot\Theta(P,a)=\Theta(UP,a)$, so under $\Theta$ the quotient map is the
second projection.  To verify the categorical universal property explicitly,
let $Y$ be a $\Z$-scheme and let
\[
 f:\operatorname{GL}_{r,\Z}\times\A_{\Z}^r\longrightarrow Y
\]
be invariant under left translation on the first factor.  Then
$f(P,a)=f(I_r,a)$, so $f$ factors through the second projection via
$a\mapsto f(I_r,a)$; this factorization is unique.  The same argument and all
the displayed formulas commute with arbitrary base change.  This proves the
torsor and quotient statements.
\end{proof}

\subsection{The root-jet model}

For a polynomial \(f(x)=\sum_i c_ix^i\), write
\[
 D^{(j)}f(x)=\sum_{i\ge j}\binom ijc_ix^{i-j}
\]
for its \(j\)-th Hasse derivative.  Suppose that over a splitting field
\[
 g(x)=\prod_{\nu=1}^s(x-\alpha_\nu)^{m_\nu},
 \qquad \sum_{\nu=1}^s m_\nu=r,
\]
with distinct centers.  Li and Yuan~\cite[Theorem~3.4]{LiYuanCyclicOrbits}
give the following model, which we record to identify precisely the input
used by the Schur factorization.

Here and below,
\(\binom{i}{j}\alpha^{\,i-j}\) is understood to be \(0\) when \(i<j\);
this is the Hasse-derivative convention and also covers \(\alpha=0\).

\begin{proposition}[Root-jet parity check]
\label{prop:root-jet-recalled}
The matrix
\begin{equation}
 H_g^{\mathrm{jet}}
 =
 \left(
 \binom{i}{j}\alpha_\nu^{\,i-j}
 \right)_{\substack{1\le\nu\le s,\ 0\le j<m_\nu\\0\le i<n}}
 \label{eq:root-jet-matrix}
\end{equation}
is a parity-check matrix of the scalar extension of
\(\mathcal C_g(n)\) to the splitting field.  It differs from
\(H_g^{\mathrm{rem}}\) by an invertible row transformation.  Hence the
code is MDS if and only if every generalized confluent Vandermonde maximal
minor of~\eqref{eq:root-jet-matrix} is nonzero.
\end{proposition}

\begin{proof}
For every \(f\),
\[
 g\mid f
 \quad\Longleftrightarrow\quad
 D^{(j)}f(\alpha_\nu)=0
 \quad(1\le\nu\le s,\ 0\le j<m_\nu).
\]
The Hermite--Hasse jet map on polynomials of degree \(<r\) is an
isomorphism, because its kernel consists of degree-\(<r\) polynomials
divisible by \(g\).  Applying this isomorphism to the remainder of \(x^i\)
shows that the jet and remainder matrices differ by an invertible left
factor.  The maximal-minor criterion is then
Proposition~\ref{prop:remainder-parity-sparse}.
\end{proof}

For a reduced index set
\[
 U=\{u_1<\cdots<u_{r-1}\}\subseteq\{1,\ldots,n-1\},
\]
put
\begin{equation}
 \delta_U(g)=\det(Q_0,Q_{u_1},\ldots,Q_{u_{r-1}}).
 \label{eq:reduced-minor}
\end{equation}
The determinantal product introduced in
\cite[Equation~(3.2) and Theorem~3.2]{LiYuanCyclicOrbits} packages the exact
test as the evaluation of
\begin{equation}
 \Ddet
 =A_0
 \prod_{\substack{U\subseteq\{1,\ldots,n-1\}\\|U|=r-1}}
 \delta_U(\mathbf g)\in R_r.
 \label{eq:determinantal-D}
\end{equation}
Section~\ref{sec:schur-system} identifies every factor in this product.

\subsection{The Grassmannian map}

Because the first $r$ columns of the universal remainder matrix form the
identity, its rows define a morphism
\begin{equation}
 \Phi_{n,r}:\A^r_{\Z}\longrightarrow\Gr(r,n),
 \qquad
 (A_0,\ldots,A_{r-1})\longmapsto
 \operatorname{rowspan}H_{\mathbf g}^{\mathrm{rem}}.
 \label{eq:grassmannian-map}
\end{equation}
It lands in the standard big cell
\[
 \Gr(r,n)^\circ
 =D\bigl(p_{\{0,\ldots,r-1\}}\bigr)
 \simeq\operatorname{Mat}_{r\times(n-r)}.
\]
In particular, the Pl\"ucker coordinate
\(p_{\{0,1,\ldots,r-1\}}\) is normalized to \(1\); all Pl\"ucker
coordinates used below are therefore concrete regular functions on this
big-cell chart, rather than homogeneous coordinates left up to a common
scalar.
For the standard Pl\"ucker embedding and big-cell coordinates, see
\cite[Lecture~6]{HarrisAG}.

\begin{proposition}[The companion coefficient space in the big cell]
\label{prop:companion-big-cell-immersion}
For every $n>r$, the factorization of $\Phi_{n,r}$ through the big cell is
a closed immersion
\begin{equation}
 \Phi_{n,r}^{\circ}:\A^r_{\Z}\hookrightarrow\Gr(r,n)^\circ.
 \label{eq:companion-big-cell-graph}
\end{equation}
Under the big-cell representation
\[
 [\,I_r\mid Q_r\ Q_{r+1}\ \cdots\ Q_{n-1}\,],
\]
its image is the graph of the polynomial map determined by the companion
recurrence over the first column
\[
 Q_r=(-A_0,-A_1,\ldots,-A_{r-1})^T.
\]
Consequently, $\Phi_{n,r}$ is a locally closed immersion into
$\Gr(r,n)$.  Its image
$\mathscr K_{n,r}:=\Phi_{n,r}^{\circ}(\A^r_{\Z})$ is closed in the big
cell, integral, and smooth of relative dimension $r$ over $\Spec\Z$.
\end{proposition}

\begin{proof}
The first column after $I_r$ recovers every coefficient through
$A_i=-(Q_r)_i$.  Each remaining column is a polynomial in these
coefficients by~\eqref{eq:remainder-recurrence}.  Thus the image in the big
cell is the graph of a polynomial map and is closed there.  The final
claims follow because this graph is isomorphic to $\A^r_{\Z}$.
\end{proof}

\begin{proposition}[Explicit complete-intersection equations]
\label{prop:companion-complete-intersection}
Put $m=n-r$ and write a point of the standard big cell as
\[
 [\,I_r\mid Z_0\ Z_1\ \cdots\ Z_{m-1}\,],
 \qquad Z_j=(z_{0,j},\ldots,z_{r-1,j})^T.
\]
The closed subscheme $\mathscr K_{n,r}$ is defined scheme-theoretically by
\begin{align}
 z_{0,j+1}-z_{0,0}z_{r-1,j}&=0,
 \label{eq:companion-CI-zero}\\
 z_{i,j+1}-z_{i-1,j}-z_{i,0}z_{r-1,j}&=0
 \quad(1\le i\le r-1),
 \label{eq:companion-CI-positive}
\end{align}
for $0\le j\le m-2$.  These $r(m-1)$ equations form a regular sequence.
Thus $\mathscr K_{n,r}$ is a smooth closed complete intersection of
codimension $r(n-r-1)$ in $\Gr(r,n)^\circ$.  When $m=1$, the equation set
is empty and $\mathscr K_{r+1,r}=\Gr(r,r+1)^\circ$.
\end{proposition}

\begin{proof}
Let
\[
 S=\Z[z_{i,j}:0\le i\le r-1,\ 0\le j\le m-1]
\]
be the coordinate ring of the big cell, and let $J$ be generated by
\eqref{eq:companion-CI-zero}--\eqref{eq:companion-CI-positive}.  Multiplication
by $x$ modulo the polynomial whose coefficients are
$A_i=-z_{i,0}$ sends a column $v=(v_0,\ldots,v_{r-1})^T$ to
\[
 (z_{0,0}v_{r-1},\ v_0+z_{1,0}v_{r-1},\ldots,
   v_{r-2}+z_{r-1,0}v_{r-1})^T.
\]
Thus the displayed equations are exactly the companion recurrences
$Z_{j+1}=T_gZ_j$.  Successively eliminating
$Z_1,\ldots,Z_{m-1}$ gives
\[
 S/J\simeq\Z[z_{0,0},\ldots,z_{r-1,0}]
 \simeq R_r,
 \qquad z_{i,0}\longmapsto-A_i.
\]
Hence $J$ is the scheme-theoretic defining ideal of $\mathscr K_{n,r}$.
Ordered by increasing $j$, each equation is monic linear in a new variable
$z_{i,j+1}$, and every successive quotient remains a polynomial ring.
The equations therefore form a regular sequence.  Their number is
$r(m-1)=r(n-r-1)$, which is the asserted codimension; the remaining claims
also follow from the displayed quotient.
\end{proof}

Let
\begin{equation}
 \Gr(r,n)^{\mathrm{unif}}
 :=\bigcap_{\substack{I\subseteq\{0,\ldots,n-1\}\\|I|=r}}D_+(p_I)
 \label{eq:uniform-Plucker-open}
\end{equation}
be the Pl\"ucker open subscheme on which every Pl\"ucker coordinate is
nonvanishing; its points represent the uniform matroid $U_{r,n}$.  This is
the Grassmannian stratum associated with the uniform matroid;
compare~\cite{GGMS}.  On the standard big cell, where
$p_{\{0,\ldots,r-1\}}=1$, the remaining normalized Pl\"ucker coordinates
are regular functions, and the uniform-matroid condition requires them to be
units.  Let
$\MDSlocus\subseteq\A^r_{\Z}$ be the open subscheme on which every maximal
minor of the universal remainder matrix is invertible.  We call
$\MDSlocus$ the \emph{coefficient MDS locus}.

\begin{proposition}[Grassmannian model of the coefficient MDS locus]
\label{prop:grassmannian-MDS}
As open subschemes over $\Spec\Z$,
\begin{equation}
 \boxed{\ \MDSlocus=
 \Phi_{n,r}^{-1}\bigl(\Gr(r,n)^{\mathrm{unif}}\bigr)\ }.
 \label{eq:MDS-pullback-uniform}
\end{equation}
Equivalently, inside the standard big cell,
\begin{equation}
 \MDSlocus\simeq
 \mathscr K_{n,r}\cap\Gr(r,n)^{\mathrm{unif}}.
 \label{eq:companion-uniform-intersection}
\end{equation}
The equality is preserved by arbitrary base change; over a field $K$, its
$K$-points are precisely the coefficient vectors for which
$\mathcal C_g(n)$ is MDS.
\end{proposition}

\begin{proof}
The pullbacks of the Pl\"ucker coordinates are the maximal minors of the
universal remainder matrix.  Both sides of
\eqref{eq:MDS-pullback-uniform} are therefore the open subscheme obtained by
requiring all these minors to be invertible.  After base change to a field,
the assertion on points is exactly~\eqref{eq:parity-MDS}.
\end{proof}

\subsection{Partitions and universal Schur coordinates}

A partition $\kappa=(\kappa_1,\ldots,\kappa_r)$ is padded with zeros when
necessary.  We write $\kappa\subseteq b^a$ if
$\kappa_1\le b$ and $\ell(\kappa)\le a$, and put
\begin{equation}
 \Part=\{\kappa:\kappa\subseteq(n-r)^{r-1}\}.
 \label{eq:partition-rectangle}
\end{equation}
For $r=1$ we adopt the conventions
\begin{equation}
 \Part=\{\varnothing\},\qquad
 \kappa(\varnothing)=\varnothing,\qquad
 \Schur{\varnothing}=1,
 \label{eq:r-one-conventions}
\end{equation}
and put $\ell(\varnothing)=|\varnothing|=0$; when
$\kappa=\varnothing$, we interpret $\kappa_1=0$.
For variables $X=(X_1,\ldots,X_r)$, let $s_\kappa(X)$ be the Schur
polynomial.  We use the bialternant formula
\begin{equation}
 s_\kappa(X)=
 \frac{\det(X_i^{\kappa_j+r-j})_{1\le i,j\le r}}
      {\det(X_i^{r-j})_{1\le i,j\le r}}
 \label{eq:bialternant}
\end{equation}
and the dual Jacobi--Trudi formula
\begin{equation}
 s_\kappa(X)=
 \det(e_{\kappa_i'-i+j}(X))_{1\le i,j\le\kappa_1},
 \label{eq:dual-JT}
\end{equation}
where $e_0=1$ and $e_j=0$ for $j<0$ or $j>r$; see, for example,
\cite[Chapter~I]{Macdonald}.

The fundamental theorem of symmetric polynomials gives an isomorphism
\begin{equation}
 R_r\longrightarrow\Z[X_1,\ldots,X_r]^{S_r},
 \qquad
 A_{r-j}\longmapsto(-1)^je_j(X).
 \label{eq:symmetric-isomorphism}
\end{equation}
For every partition of length at most $r$, define
$\Schur{\kappa}\in R_r$ to be the inverse image of $s_\kappa(X)$.  If
$\Lambda_g=(\lambda_1,\ldots,\lambda_r)$ is the root multiset of a concrete
$g$ in an algebraic closure, then
\begin{equation}
 \Schur{\kappa}(g)=s_\kappa(\Lambda_g).
 \label{eq:Schur-evaluation}
\end{equation}
This value belongs to the coefficient field, even when the roots do not.
Define the universal Schur product
\begin{equation}
 \DSch:=A_0\prod_{\kappa\in\Part}\Schur{\kappa}\in R_r.
 \label{eq:Schur-D}
\end{equation}
In particular, $\mathscr D^{\mathrm{Sch}}_{n,1}=A_0$.

\section{The rectangular Schur--Pl\"ucker system}
\label{sec:schur-system}

\subsection{Index sets and the partition rectangle}

For $r>1$ and a reduced index set
$U=\{u_1<\cdots<u_{r-1}\}$ define
\begin{equation}
 \kappa(U)=
 \bigl(u_{r-1}-(r-1),u_{r-2}-(r-2),\ldots,u_1-1,0\bigr).
 \label{eq:kappa-U}
\end{equation}
For $r=1$ use~\eqref{eq:r-one-conventions}.

\begin{lemma}[Index-set/partition bijection]
\label{lem:index-partition}
The map $U\mapsto\kappa(U)$ is a bijection from the $(r-1)$-subsets of
$\{1,\ldots,n-1\}$ to $\Part$.  Its inverse is
\begin{equation}
 u_j=j+\kappa_{r-j},\qquad 1\le j\le r-1.
 \label{eq:partition-index-inverse}
\end{equation}
In particular, $|\Part|=\binom{n-1}{r-1}$.
\end{lemma}

\begin{proof}
The case $r=1$ is immediate from~\eqref{eq:r-one-conventions}; assume
$r>1$.  
The strict inequalities among the $u_j$ make the displayed parts weakly
decreasing.  The largest part is at most
$(n-1)-(r-1)=n-r$, and the last part is zero.  Conversely,
\eqref{eq:partition-index-inverse} gives a strictly increasing subset of
$\{1,\ldots,n-1\}$ and is visibly inverse to~\eqref{eq:kappa-U}.
\end{proof}

\subsection{Reduced minors as Schur coordinates}

\begin{theorem}[Schur coordinates of the reduced minors]
\label{thm:reduced-Schur}
For every field, every monic $g$ of degree $r$, and every reduced index set
$U$,
\begin{equation}
 \delta_U(g)=\Schur{\kappa(U)}(g)
 =s_{\kappa(U)}(\lambda_1,\ldots,\lambda_r).
 \label{eq:reduced-Schur-identity}
\end{equation}
The identity is characteristic-free and requires no squarefreeness
assumption.
\end{theorem}

\begin{proof}
First suppose that the roots are pairwise distinct and let
\[
 B=(\lambda_i^j)_{\substack{1\le i\le r\\0\le j\le r-1}}.
\]
Evaluation of the remainder identity at all roots gives
\[
 (\lambda_1^m,\ldots,\lambda_r^m)^T=BQ_m.
\]
Therefore
\[
 \det(1,\lambda_i^{u_1},\ldots,\lambda_i^{u_{r-1}})_{1\le i\le r}
 =\det(B)\,\delta_U(g).
\]
For $\kappa(U)$, the numerator exponents in~\eqref{eq:bialternant} are
$u_{r-1},\ldots,u_1,0$.  Reversing the columns in both numerator and
denominator identifies the determinant ratio with $\delta_U(g)$.

For the universal argument, take algebraically independent variables
$X_1,\ldots,X_r$ and
$G_X(x)=\prod_i(x-X_i)$.  The recurrence defining the remainder columns
shows that $\delta_U(G_X)$ is a symmetric integral polynomial in the $X_i$.
The squarefree calculation over the fraction field gives
\[
 \delta_U(G_X)=s_{\kappa(U)}(X_1,\ldots,X_r)
\]
as an identity of integral polynomials.  Specializing the $X_i$, allowing
coincidences, and reducing modulo an arbitrary prime proves the theorem.
\end{proof}

\subsection{All Pl\"ucker coordinates and the MDS boundary}

\begin{theorem}[Rectangular Schur--Pl\"ucker theorem]
\label{thm:rectangular-Schur-Pluecker}
Let
$I=\{i_0<i_1<\cdots<i_{r-1}\}\subseteq\{0,\ldots,n-1\}$ and set
$U_I=\{i_1-i_0,\ldots,i_{r-1}-i_0\}$.  Then
\begin{equation}
 p_I\bigl(\Phi_{n,r}(g)\bigr)
 =\det(Q_{i_0},\ldots,Q_{i_{r-1}})
 =\bigl((-1)^ra_0\bigr)^{i_0}\Schur{\kappa(U_I)}(g).
 \label{eq:all-Pluecker}
\end{equation}
Moreover, the determinantal and Schur products agree identically:
\begin{equation}
 \boxed{\ \Ddet=\DSch\ }.
 \label{eq:det-Schur-identity}
\end{equation}
We henceforth denote their common value by $\Duniv$; thus
\begin{equation}
 \Duniv=A_0\prod_{\kappa\in\Part}\Schur{\kappa}.
 \label{eq:canonical-Schur-product}
\end{equation}
and, over every field,
\begin{equation}
 \mathcal C_g(n)\text{ is MDS}
 \Longleftrightarrow
 a_0\ne0\ \text{and}\quad
 \Schur{\kappa}(g)\ne0\quad(\kappa\in\Part).
 \label{eq:Schur-MDS-test}
\end{equation}
\end{theorem}

\begin{proof}
Since $Q_j=T_g^jQ_0$, every column indexed by $i_j$ can be written as
$T_g^{i_0}Q_{i_j-i_0}$.  Factoring $T_g^{i_0}$ from the determinant gives
\[
 \det(Q_{i_0},\ldots,Q_{i_{r-1}})
 =\det(T_g)^{i_0}\delta_{U_I}(g).
\]
Use~\eqref{eq:companion-shift} and
Theorem~\ref{thm:reduced-Schur}.  Lemma~\ref{lem:index-partition} then
transforms~\eqref{eq:determinantal-D} into
\eqref{eq:canonical-Schur-product}.  If all reduced factors and $a_0$ are
nonzero, formula~\eqref{eq:all-Pluecker} makes every maximal minor nonzero.
The converse follows from the reduced minors and, for example, the minor
$\det(Q_1,\ldots,Q_r)=\det(T_g)=(-1)^ra_0$.  Proposition~\ref{prop:remainder-parity-sparse}
completes the proof.
\end{proof}

\begin{corollary}[The universal MDS boundary divisor]
\label{cor:universal-MDS-divisor}
As an open subscheme of coefficient space,
\begin{equation}
 \MDSlocus=D(\Duniv)\subseteq\A^r_{\Z}.
 \label{eq:MDS-principal-open}
\end{equation}
Since $R_r$ is a domain and $\Duniv\ne0$, the polynomial $\Duniv$ is a
non-zero-divisor.  It therefore defines the principal effective Cartier
divisor
\begin{equation}
 \Delta_{n,r}:=\operatorname{div}(\Duniv)
 \quad\text{on }\A^r_{\Z}.
 \label{eq:MDS-Cartier-divisor}
\end{equation}
Its complement is $\MDSlocus$, and its support is the coefficient-space
non-MDS locus:
\begin{equation}
 |\Delta_{n,r}|=\A^r_{\Z}\setminus\MDSlocus
 =V(A_0)\cup
 \bigcup_{\kappa\in\Part}V(\Schur{\kappa}).
 \label{eq:MDS-boundary-union}
\end{equation}
More precisely, as effective Cartier divisors,
\begin{equation}
 \Delta_{n,r}=\operatorname{div}(A_0)
 +\sum_{\kappa\in\Part}\operatorname{div}(\Schur{\kappa}).
 \label{eq:MDS-boundary-divisor-sum}
\end{equation}
The product~\eqref{eq:canonical-Schur-product} is a canonical
Schur-coordinate factorization; no irreducibility assertion about the
individual $\Schur{\kappa}$ is intended, nor are the displayed factors
asserted to be pairwise coprime.
\end{corollary}

\begin{proposition}[The full Pl\"ucker product and its boundary multiplicities]
\label{prop:full-Pluecker-product}
Using the big-cell normalization
$p_{\{0,\ldots,r-1\}}=1$, let
\[
 \mathcal B^{\mathrm{Pl}}_{n,r}
 :=\sum_{\substack{I\subseteq\{0,\ldots,n-1\}\\ |I|=r}}
   \operatorname{div}(p_I)
 \quad\text{on }\Gr(r,n)^\circ,
\]
and put
\[
 F^{\mathrm{Pl}}_{n,r}
 :=\prod_{\substack{I\subseteq\{0,\ldots,n-1\}\\ |I|=r}}
     \bigl(p_I\circ\Phi_{n,r}\bigr)\in R_r,
 \qquad
 E_{n,r}:=\binom n{r+1}.
\]
Then
\begin{equation}
 \boxed{\
 F^{\mathrm{Pl}}_{n,r}
 =(-1)^{rE_{n,r}}A_0^{E_{n,r}}
  \prod_{\kappa\in\Part}
  \Schur{\kappa}^{\,n-r+1-\kappa_1}
 }.
 \label{eq:full-Pluecker-product}
\end{equation}
Its ordinary total degree is $rE_{n,r}$.  Moreover,
\begin{equation}
\begin{aligned}
 (\Phi_{n,r}^{\circ})^*\mathcal B^{\mathrm{Pl}}_{n,r}
 &=\operatorname{div}(F^{\mathrm{Pl}}_{n,r})\\
 &=E_{n,r}\operatorname{div}(A_0)
   +\sum_{\kappa\in\Part}(n-r+1-\kappa_1)
     \operatorname{div}(\Schur{\kappa})\\
 &=\Delta_{n,r}+(E_{n,r}-1)\operatorname{div}(A_0)
   +\sum_{\kappa\in\Part}(n-r-\kappa_1)
     \operatorname{div}(\Schur{\kappa}).
\end{aligned}
\label{eq:full-Pluecker-divisor}
\end{equation}
In particular,
\begin{equation}
 D(F^{\mathrm{Pl}}_{n,r})=\MDSlocus,
 \qquad
 \bigl|\operatorname{div}(F^{\mathrm{Pl}}_{n,r})\bigr|
 =|\Delta_{n,r}|.
 \label{eq:full-Pluecker-support}
\end{equation}
Thus the full Pl\"ucker product and the canonical Schur product define the
same open set and the same boundary support, although their Cartier
multiplicities are generally different.
\end{proposition}

\begin{proof}
Multiplying~\eqref{eq:all-Pluecker} over all $r$-subsets $I$ gives an
$A_0$-exponent $\sum_I i_0$ and the sign $(-1)^{r\sum_I i_0}$.  Since
\[
 \sum_I i_0
 =\sum_{t=1}^{n-r}\binom{n-t}{r}
 =\binom n{r+1}=E_{n,r},
\]
these are the first two factors in~\eqref{eq:full-Pluecker-product}.  For a
fixed $\kappa$, the corresponding reduced set has
$u_{r-1}=\kappa_1+r-1$ and may be translated by precisely
$0\le i_0\le n-1-u_{r-1}$.  Hence $\Schur{\kappa}$ occurs
$n-r+1-\kappa_1$ times.

By Lemma~\ref{lem:Schur-degree},
$\deg\Schur{\kappa}=\kappa_1$.  Also
$i_0+\kappa_1(U_I)=i_{r-1}-(r-1)$, and therefore
\[
 \deg F^{\mathrm{Pl}}_{n,r}
 =\sum_{j=r-1}^{n-1}(j-r+1)\binom j{r-1}
 =r\sum_{j=r}^{n-1}\binom jr
 =r\binom n{r+1}.
\]
Taking principal divisors gives~\eqref{eq:full-Pluecker-divisor}; every
factor exponent in~\eqref{eq:full-Pluecker-product} is positive, and
\eqref{eq:full-Pluecker-support} follows from
Corollary~\ref{cor:universal-MDS-divisor}.
\end{proof}

\begin{proposition}[The constant-coefficient component]
\label{prop:constant-coefficient-component}
Let
\[
 H_0:=V(A_0)\subseteq\A^r_{\Z}.
\]
For every $n>r\ge1$, the integral prime divisor $H_0$ is an irreducible
component of $|\Delta_{n,r}|$, and its coefficient in the effective Cartier
divisor $\Delta_{n,r}$ is one.  Equivalently, $\Delta_{n,r}$ is generically
reduced along $H_0$.  More precisely,
\begin{equation}
 A_0\nmid\Schur{\kappa}
 \qquad(\kappa\in\Part),
 \label{eq:A0-not-Schur-factor}
\end{equation}
so $H_0$ is contained in none of the Schur divisors
$V(\Schur{\kappa})$.  The same assertions hold in every field fiber.
\end{proposition}

\begin{proof}
Under the symmetric-polynomial isomorphism
\eqref{eq:symmetric-isomorphism}, quotienting by $A_0$ amounts to setting
$e_r=0$, equivalently to the stable specialization obtained by setting the
last root variable equal to zero.  Since every $\kappa\in\Part$ has
$\ell(\kappa)\le r-1$,
\[
 s_\kappa(X_1,\ldots,X_{r-1},0)
 =s_\kappa(X_1,\ldots,X_{r-1})\ne0.
\]
By the unitriangular Kostka expansion, its monomial-symmetric expansion
contains $m_\kappa$ with coefficient one~\cite[Chapter~I, \S6]{Macdonald},
so this specialization remains nonzero over every field.  Thus no
$\Schur{\kappa}$ is divisible by $A_0$.  If
$F=\prod_{\kappa\in\Part}\Schur{\kappa}$, then $F$ is a unit in the local
ring at the height-one prime $(A_0)$, while $\Duniv=A_0F$.  Hence
$\operatorname{ord}_{H_0}(\Duniv)=1$.  For $r=1$, only the empty Schur
factor occurs, and the conclusion is immediate.
\end{proof}

\subsection{Coefficient formulas and confluent specialization}

Vieta's formulas give $e_j(\Lambda_g)=(-1)^ja_{r-j}$, so the dual
Jacobi--Trudi determinant~\eqref{eq:dual-JT} computes every
$\Schur{\kappa}(g)$ directly from the coefficients.  There is also a useful
recurrence through the complete homogeneous functions $h_m(\Lambda_g)$:
\begin{equation}
 \sum_{m\ge0}h_m(\Lambda_g)T^m
 =\frac{1}{1+a_{r-1}T+a_{r-2}T^2+\cdots+a_0T^r},
 \label{eq:h-generating}
\end{equation}
so $h_0=1$, $h_m=0$ for $m<0$, and
\begin{equation}
 h_m=-\sum_{j=1}^{\min(m,r)}a_{r-j}h_{m-j}.
 \label{eq:h-recurrence}
\end{equation}
The ordinary Jacobi--Trudi identity then gives a second coefficient-only
calculation of all factors.

Suppose now that
\begin{equation}
 g(x)=\prod_{\nu=1}^s(x-\alpha_\nu)^{m_\nu},
 \qquad
 \alpha_\mu\ne\alpha_\nu\ (\mu\ne\nu),
 \qquad
 \sum_{\nu=1}^s m_\nu=r.
 \label{eq:multiplicity-factorization}
\end{equation}
Let $0\le i_1<\cdots<i_r$ and put $I=\{i_1,\ldots,i_r\}$.  Let
$F_I(\alpha_1,\ldots,\alpha_s)$ be the Hasse-derivative confluent
Vandermonde determinant whose rows are indexed by $(\nu,j)$,
$0\le j<m_\nu$, and whose entries are
$\binom{i_t}{j}\alpha_\nu^{i_t-j}$, with the same zero convention when
$i_t<j$.  Put
\begin{equation}
 \lambda(I)=
 (i_r-(r-1),i_{r-1}-(r-2),\ldots,i_2-1,i_1).
 \label{eq:lambda-I}
\end{equation}

\begin{proposition}[Confluent Schur factorization]
\label{prop:confluent-Schur}
Order the root blocks by $\alpha_1,\ldots,\alpha_s$ and the rows in the
$\nu$-th block by $j=0,\ldots,m_\nu-1$.  Then
\begin{equation}
 F_I(\alpha_1,\ldots,\alpha_s)
 =\prod_{\mu<\nu}
 (\alpha_\nu-\alpha_\mu)^{m_\mu m_\nu}
 s_{\lambda(I)}
 (\alpha_1^{[m_1]},\ldots,\alpha_s^{[m_s]}).
 \label{eq:confluent-Schur}
\end{equation}
For $I=\{0,u_1,\ldots,u_{r-1}\}$, the last factor is
$s_{\kappa(U)}(\alpha_1^{[m_1]},\ldots,\alpha_s^{[m_s]})$.
The displayed formula is an identity in
\(\Z[\alpha_1,\ldots,\alpha_s]\), and therefore specializes to every
characteristic.
\end{proposition}

\begin{proof}
Work first over a characteristic-zero fraction field.  Introduce pairwise
distinct variables
\[
 X_{\nu,0},\ldots,X_{\nu,m_\nu-1}
 \qquad(1\le\nu\le s),
\]
ordered blockwise in the same order as the rows of the confluent matrix.
For distinct variables, the ordinary generalized bialternant identity
identifies the quotient of the alternant with exponent set $I$ by the
ordinary Vandermonde alternant with $s_{\lambda(I)}(X)$.

For each block divide both alternants by its internal Vandermonde
\[
 V_\nu=\prod_{0\le a<b<m_\nu}(X_{\nu,b}-X_{\nu,a}),
\]
replace the rows successively by divided-difference rows, and let
$X_{\nu,j}\to\alpha_\nu$.  The $j$-th divided-difference row then tends to
the Hasse row
\[
 D^{(j)}(x^{i_t})(\alpha_\nu)
 =\binom{i_t}{j}\alpha_\nu^{i_t-j}.
\]
Thus the numerator tends to $F_I(\alpha_1,\ldots,\alpha_s)$.  After the
internal Vandermonde factors have been removed, the denominator tends to
\[
 \prod_{\mu<\nu}(\alpha_\nu-\alpha_\mu)^{m_\mu m_\nu}.
\]
Consequently the limiting bialternant quotient gives exactly
\eqref{eq:confluent-Schur}.  This calculation also fixes the displayed sign
under the chosen block order and shows directly why the Hasse normalization
introduces no factorial factors; compare
\cite[Theorem~1.1]{GonzalezSerranoMaximenko}.  Both sides are integral
polynomials, so the identity extends to
$\Z[\alpha_1,\ldots,\alpha_s]$, and hence to every characteristic.
\end{proof}

\section{Arithmetic of the universal MDS boundary divisor}
\label{sec:global-arithmetic}

\subsection{Primitivity and exact degrees}

Give the coefficient variables the root-weight grading
\begin{equation}
 \rwt(A_{r-j})=j,
 \qquad 1\le j\le r.
 \label{eq:root-weight}
\end{equation}
Under the root dilation
\[
 g(x)\longmapsto g_c(x):=c^rg(x/c),
\]
the coefficient $A_{r-j}$ is multiplied by $c^j$.

\begin{lemma}[Degree of a Schur coordinate]
\label{lem:Schur-degree}
For every partition $\kappa$ of length at most $r$,
\begin{equation}
 \deg_{A_0,\ldots,A_{r-1}}\Schur{\kappa}=\kappa_1,
 \qquad
 \rwt(\Schur{\kappa})=|\kappa|.
 \label{eq:one-Schur-degree}
\end{equation}
\end{lemma}

\begin{proof}
If $\kappa=\varnothing$, both assertions are immediate.  Assume
$\kappa\ne\varnothing$, put $b=\kappa_1$, and write
$\kappa'=(c_1,\ldots,c_b)$.  In the dual Jacobi--Trudi determinant, each
nonzero term has ordinary degree at most $b$.  For any partition $\mu$,
write
\[
 e_{\mu'}:=\prod_{j=1}^{\mu_1}e_{\mu'_j}.
\]
Inverse Kostka triangularity, with respect to the dominance order on
partitions of $|\kappa|$, gives
\[
 s_\kappa=e_{\kappa'}
 +\sum_{\mu\lhd\kappa}d_{\kappa\mu}e_{\mu'},
 \qquad d_{\kappa\mu}\in\Z;
\]
see~\cite[Chapter~I, \S6]{Macdonald}.  Here $\mu\lhd\kappa$ means that
$\mu$ is strictly dominated by $\kappa$.  After specialization to $r$
variables, terms with $\ell(\mu)>r$ vanish because they contain
$e_{\mu'_1}=e_{\ell(\mu)}=0$, and are omitted.  Distinct surviving
partitions give distinct monomials in the algebraically independent
variables $e_1,\ldots,e_r$.  Hence
$e_{\kappa'}=e_{c_1}\cdots e_{c_b}$ occurs with coefficient one and cannot
cancel.
Since every $c_i>0$, this monomial has ordinary degree $b$, which proves the
first equality.  Finally, $e_j=(-1)^jA_{r-j}$ has root weight $j$, and
homogeneity gives $\rwt(\Schur{\kappa})=|\kappa|$.
\end{proof}

\begin{theorem}[Primitivity and exact degrees]
\label{thm:degree-primitivity}
For every $n>r\ge1$, the polynomial $\Duniv$ is primitive in
$\Z[A_0,\ldots,A_{r-1}]$ and remains nonzero after reduction modulo every
prime.  Put
\begin{align}
 W_{n,r}
 &:=\frac{(r-1)(n-r)}{2}\binom{n-1}{r-1},
 \label{eq:Wnr}\\
 \delta_{n,r}
 &:=1+\frac{(r-1)(n-r)}{r}\binom{n-1}{r-1}.
 \label{eq:delta-nr}
\end{align}
Then
\begin{align}
 \sum_{\kappa\in\Part}|\kappa|&=W_{n,r},
 \label{eq:sum-sizes}\\
 \deg\Duniv&=\delta_{n,r},
 \label{eq:D-total-degree}\\
 \rwt(\Duniv)&=r+W_{n,r}.
 \label{eq:D-weighted-degree}
\end{align}
\end{theorem}

\begin{proof}
The monomial-symmetric expansion of $s_\kappa$ contains $m_\kappa$ with
coefficient one.  Thus $s_\kappa$ does not vanish after reduction modulo any
prime.  Under the integral isomorphism~\eqref{eq:symmetric-isomorphism}, this
implies that every $\Schur{\kappa}$ is primitive.  Gauss' lemma and
\eqref{eq:canonical-Schur-product} give the primitivity of $\Duniv$.

Set $a=r-1$ and $b=n-r$.  Complementation in the $a\times b$ rectangle
pairs every partition of size $j$ with one of size $ab-j$, so
\[
 \sum_{\kappa\subseteq b^a}|\kappa|
 =\frac{ab}{2}\binom{a+b}{a},
\]
which is~\eqref{eq:sum-sizes}.  By Lemma~\ref{lem:Schur-degree}, the
ordinary degree of the Schur product is $\sum_\kappa\kappa_1$.  The number
of partitions in $j^a$ is $\binom{a+j}{a}$, and therefore
\begin{align*}
 \sum_{\kappa\subseteq b^a}\kappa_1
 &=\sum_{j=1}^b
 \left(\binom{a+b}{a}-\binom{a+j-1}{a}\right)\\
 &=b\binom{a+b}{a}-\binom{a+b}{a+1}
 =\frac{ab}{a+1}\binom{a+b}{a}.
\end{align*}
Adding the external factor $A_0$ gives~\eqref{eq:D-total-degree}.  Its
root weight is $r$, so~\eqref{eq:D-weighted-degree} follows from
\eqref{eq:sum-sizes}.
\end{proof}

\begin{lemma}[Flatness criterion for hypersurfaces over $\Z$]
\label{lem:integral-hypersurface-flatness}
Let $F\in\Z[X_1,\ldots,X_m]$ have nonzero reduction modulo every prime.
Then
\[
 V(F)\longrightarrow\Spec\Z
\]
is flat.
\end{lemma}

\begin{proof}
Put $R=\Z[X_1,\ldots,X_m]$.  If $p[f]=0$ in $R/(F)$ for a prime $p$, then
$pf=Fh$ for some $h\in R$.  Reduction modulo $p$ gives
$\overline F\,\overline h=0$ in the domain
$\F_p[X_1,\ldots,X_m]$.  Since $\overline F\ne0$, one has
$\overline h=0$.  Thus $h\in pR$, say $h=ph_1$;
cancelling $p$ in $R$ gives $f=Fh_1$.  Thus $R/(F)$ is torsion-free over
$\Z$.  Since $\Z$ is a PID, every torsion-free $\Z$-module is flat.
\end{proof}

\begin{corollary}[Flatness and constant degree of the universal boundary]
\label{cor:boundary-flatness}
The structural morphism
$\Delta_{n,r}\to\Spec\Z$ is flat.  For every field $k$, its fiber is the
nonzero effective Cartier hypersurface
\[
 V\bigl((\Duniv)_k\bigr)\subseteq\A_k^r
\]
of exact total degree $\delta_{n,r}$.  Moreover,
$(\MDSlocus)_k=D((\Duniv)_k)$ is a nonempty smooth geometrically integral
open subscheme of $\A_k^r$.  In particular, the structural morphism
\[
 \MDSlocus\longrightarrow\Spec\Z
\]
is smooth and surjective of relative dimension $r$, with geometrically
integral fibers.
\end{corollary}

\begin{proof}
The proof of Lemma~\ref{lem:Schur-degree} exhibits in the ordinary top-degree
part of every $\Schur{\kappa}$ a monomial with coefficient $1$ in the
elementary symmetric variables, hence with coefficient $\pm1$ in the
coefficient variables.  Consequently
\[
 \deg(\Schur{\kappa}\bmod p)=\kappa_1
\]
for every prime $p$.  Degrees add in a polynomial ring over a field, so
$\deg((\Duniv)_k)=\delta_{n,r}$ for every field $k$; in particular, every
fiber is a nonzero effective Cartier hypersurface.

Flatness follows from Lemma~\ref{lem:integral-hypersurface-flatness}.
Finally, the principal open defined by a nonzero polynomial in
affine space is nonempty; as an open subscheme of $\A_k^r$, it is smooth and
geometrically integral.  Since $\MDSlocus$ is open in $\A^r_{\Z}$, its
structural morphism is smooth of relative dimension $r$.  Every
residue-field fiber is nonempty, so the morphism is surjective.
\end{proof}

\begin{proposition}[Scaling action]
\label{prop:scaling-action}
For $c\in K^*$ put $g_c(x)=c^rg(x/c)$.  Then
\begin{equation}
 \Schur{\kappa}(g_c)=c^{|\kappa|}\Schur{\kappa}(g),
 \qquad
 \Duniv(g_c)=c^{r+W_{n,r}}\Duniv(g).
 \label{eq:scaling-action}
\end{equation}
Hence $\MDSlocus$ is invariant under the $\Gm$-action with the weights
in~\eqref{eq:root-weight}.
\end{proposition}

\begin{proof}
The roots of $g_c$ are $c\lambda_1,\ldots,c\lambda_r$.  Apply the ordinary
homogeneity of the Schur polynomial and then multiply the factors in
\eqref{eq:canonical-Schur-product}.
\end{proof}

\begin{proposition}[Reciprocal symmetry]
\label{prop:reciprocal-symmetry}
Let $g$ be monic with constant coefficient $a_0\ne0$, put
$\rho=(-1)^ra_0$, and define
\[
 g^\#(x)=a_0^{-1}x^rg(x^{-1}).
\]
For $\kappa=(\kappa_1,\ldots,\kappa_r)\in\Part$, where
$\kappa_r=0$, put
\[
 \kappa^\vee=(\kappa_1-\kappa_r,\kappa_1-\kappa_{r-1},
 \ldots,\kappa_1-\kappa_1).
\]
Then $\kappa\mapsto\kappa^\vee$ is an involution of $\Part$, and
\begin{align}
 \Schur{\kappa}(g^\#)
 &=\rho^{-\kappa_1}\Schur{\kappa^\vee}(g),
 \label{eq:reciprocal-Schur}\\
 \Duniv(g^\#)
 &=(-1)^{r(\delta_{n,r}-1)}
   a_0^{-(\delta_{n,r}+1)}\Duniv(g).
 \label{eq:reciprocal-D}
\end{align}
In particular, $g$ defines an MDS code of length $n$ if and only if
$g^\#$ does.
\end{proposition}

\begin{proof}
The roots of $g^\#$ are the inverses of the roots of $g$.  The standard
complement identity for Schur polynomials gives
\[
 s_\kappa(\lambda_1^{-1},\ldots,\lambda_r^{-1})
 =(\lambda_1\cdots\lambda_r)^{-\kappa_1}
  s_{\kappa^\vee}(\lambda_1,\ldots,\lambda_r).
\]
Since $\lambda_1\cdots\lambda_r=\rho$, this proves
\eqref{eq:reciprocal-Schur}.  Multiplying over $\Part$, using
$\sum_{\kappa\in\Part}\kappa_1=\delta_{n,r}-1$ and the fact that the
constant coefficient of $g^\#$ is $a_0^{-1}$, gives
\eqref{eq:reciprocal-D}.  The last assertion follows from the Schur MDS
test.
\end{proof}

\begin{remark}[Code equivalences behind the symmetries]
\label{rem:code-equivalences}
The MDS invariances above also follow from explicit code equivalences, while
the Schur identities describe their finer action factor by factor.  For a
length-$n$ polynomial $f(x)=\sum_{i=0}^{n-1}f_ix^i$ and $c\in K^*$, define
\[
 \Theta_c(f)(x)=c^{n-1}f(x/c).
\]
On coefficient vectors this is the invertible diagonal transformation
\[
 (f_0,f_1,\ldots,f_{n-1})\longmapsto
 (c^{n-1}f_0,c^{n-2}f_1,\ldots,f_{n-1}).
\]
If $f=ug$, then
\[
 \Theta_c(f)=\bigl(c^{n-r-1}u(x/c)\bigr)g_c(x),
\]
and therefore
\begin{equation}
 \Theta_c\bigl(\mathcal C_g(n)\bigr)=\mathcal C_{g_c}(n).
 \label{eq:dilation-code-equivalence}
\end{equation}
Likewise, coordinate reversal
\[
 R_n(f)(x)=x^{n-1}f(x^{-1})
\]
satisfies
\[
 R_n(ug)=a_0\bigl(x^{n-r-1}u(x^{-1})\bigr)g^\#(x),
\]
so
\begin{equation}
 R_n\bigl(\mathcal C_g(n)\bigr)=\mathcal C_{g^\#}(n).
 \label{eq:reciprocal-code-equivalence}
\end{equation}
Thus dilation gives a monomial equivalence and reciprocal inversion gives a
permutation equivalence.  In particular, the paired codes have the same
complete weight distribution and the same GRS or non-GRS status.  The Schur
identities above refine these elementary code equivalences by describing
their action factor by factor on the universal MDS divisor.  Together the
two operations generate a $\Gm\rtimes C_2$ symmetry on $D(A_0)$, where the
nontrivial element of $C_2$ sends $c\in\Gm$ to $c^{-1}$.
\end{remark}

\begin{proposition}[Splitting off the scaling direction]
\label{prop:scaling-quotient}
Assume $r\ge2$ and $n>r$, and define the normalized slice
\begin{equation}
 \mathcal M_{n,r}^{(1)}
 :=\MDSlocus\cap V(A_{r-1}-1).
 \label{eq:normalized-MDS-slice}
\end{equation}
Then $\Schur{(1)}=-A_{r-1}$ is a unit on $\MDSlocus$, and the weighted
scaling action induces a scheme isomorphism
\begin{equation}
 \Gm\times\mathcal M_{n,r}^{(1)}
 \xrightarrow{\ \sim\ }\MDSlocus,
 \qquad
 (t,b)\longmapsto t\cdot b.
 \label{eq:MDS-scaling-product}
\end{equation}
Its inverse has first coordinate $t=A_{r-1}$ and quotient coordinates
\begin{equation}
 B_j=\frac{A_{r-j}}{A_{r-1}^{\,j}},
 \qquad 2\le j\le r.
 \label{eq:scaling-quotient-coordinates}
\end{equation}
Thus the action is free and $\mathcal M_{n,r}^{(1)}$ represents the fpqc
sheaf quotient of $\MDSlocus$ by $\Gm$.  In particular, for every finite
field $\F_q$,
\begin{equation}
 \#\MDSlocus(\F_q)
 =(q-1)\#\mathcal M_{n,r}^{(1)}(\F_q).
 \label{eq:MDS-count-divisibility}
\end{equation}
\end{proposition}

\begin{proof}
The partition $(1)$ belongs to $(n-r)^{r-1}$, and
$\Schur{(1)}=e_1=-A_{r-1}$.  The Schur MDS test therefore places
$\MDSlocus$ inside $D(A_{r-1})$.  On this principal open put
$t=A_{r-1}$, $B_1=1$, and use~\eqref{eq:scaling-quotient-coordinates}.  The
change of variables and its inverse are
\[
 A_{r-j}=t^jB_j\quad(1\le j\le r),
 \qquad
 t=A_{r-1},\quad B_j=A_{r-j}/A_{r-1}^j.
\]
It identifies $D(A_{r-1})$ with
 $\Gm\times\A_{\Z}^{r-1}$.  Since $\Duniv$ is root-weight homogeneous of
 weight $r+W_{n,r}$,
\[
 \Duniv(t^rB_r,\ldots,t^2B_2,t)
 =t^{\,r+W_{n,r}}\Duniv(B_r,\ldots,B_2,1).
\]
Hence its nonvanishing depends only on the normalized coordinates, proving
\eqref{eq:MDS-scaling-product}.  The quotient and counting statements follow
immediately.
\end{proof}

\subsection{Finite-field estimates and the exact GRS subtraction}

\begin{corollary}[Explicit MDS estimate]
\label{cor:MDS-count}
For every finite field $\F_q$,
\begin{equation}
 \#\MDSlocus(\F_q)
 =\#\{g\in\F_q[x]:g\text{ monic of degree }r,
                 \ \mathcal C_g(n)\text{ MDS}\}
 \ge q^r-\delta_{n,r}q^{r-1}.
 \label{eq:MDS-count}
\end{equation}
In particular, $q>\delta_{n,r}$ guarantees an MDS member in every
characteristic, and for fixed $(n,r)$ the MDS density is
$1-O_{n,r}(q^{-1})$.  More generally, if $r$ is fixed and
$n=n(q)>r$ satisfies $n(q)^r=o(q)$, then
\begin{equation}
 \frac{\#\mathcal M_{n(q),r}(\F_q)}{q^r}
 =1-O_r\!\left(\frac{n(q)^r}{q}\right)\longrightarrow1.
 \label{eq:growing-full-density}
\end{equation}
\end{corollary}

\begin{proof}
Theorem~\ref{thm:degree-primitivity} makes $\Duniv$ a nonzero polynomial of
degree $\delta_{n,r}$ over every $\F_q$.  The standard finite-field zero
bound gives at most $\delta_{n,r}q^{r-1}$ zeros
\cite[Chapter~6]{LidlNiederreiter}.  Use~\eqref{eq:MDS-principal-open}.
The growing-length assertion follows from
$\delta_{n,r}=O_r(n^r)$ for fixed $r$.
\end{proof}

For a positive integer $M$, define
\begin{equation}
 \mathcal E_M(n)=\#\{t\in C_M:\ord(t)\ge n\}
 =\sum_{\substack{e\mid M\\e\ge n}}\varphi(e),
 \label{eq:EMn}
\end{equation}
where $C_M$ is a cyclic group of order $M$.  Here $\varphi$ is Euler's
totient function, and $\mathbf 1_{\{\mathcal P\}}$ equals $1$ when the
property $\mathcal P$ holds and $0$ otherwise.  We retain the set-theoretic
notation $\GRSlocus(\F_q)$ introduced above.

\begin{corollary}[Explicit MDS non-GRS lower bound]
\label{cor:non-GRS-exact-subtraction}
Assume $r\ge3$ and $n\ge r+3$, and put $p=\operatorname{char}\F_q$.
Li and Yuan~\cite[Theorem~6.14]{LiYuanCyclicOrbits} give
\begin{equation}
 \#\GRSlocus(\F_q)
 =(q-1)\mathbf 1_{\{n\le p\}}
 +\frac{q-1}{2}\bigl(\mathcal E_{q-1}(n)+\mathcal E_{q+1}(n)\bigr).
 \label{eq:exact-GRS-count}
\end{equation}
Consequently,
\begin{align}
 &\#\{g:\mathcal C_g(n)\text{ is MDS and non-GRS}\}
 \notag\\
 &\quad\ge
 \max\!\left\{0,\,
 \begin{aligned}
 &q^r-\delta_{n,r}q^{r-1}
 -(q-1)\mathbf 1_{\{n\le p\}}\\
 &\quad-\frac{q-1}{2}
 \bigl(\mathcal E_{q-1}(n)+\mathcal E_{q+1}(n)\bigr)
 \end{aligned}
 \right\}.
 \label{eq:non-GRS-exact-lower}
\end{align}
\end{corollary}

\begin{proof}
Subtract the exact number of GRS members in~\eqref{eq:exact-GRS-count} from
the MDS lower bound~\eqref{eq:MDS-count}.  Formula
\eqref{eq:non-GRS-exact-lower} is an explicit quantitative refinement of
the asymptotic genericity theorem
\cite[Theorem~6.15]{LiYuanCyclicOrbits}; density one itself is not being
reasserted as a new result.  Taking the maximum with zero uses the
tautological nonnegativity of the count.
\end{proof}

\subsection{Application to the GRS classification: a split non-GRS bound}

We next derive, directly in remainder coordinates, the low-degree obstruction
needed below.  The calculation is self-contained once the
rational-normal-curve rigidity lemma
\cite[Lemma~4.3]{LiYuanCyclicOrbits} is taken as input, and it keeps the MDS
and GRS tests logically separate.

\begin{proposition}[A low-degree coefficient obstruction for GRS type]
\label{prop:Omega-obstruction}
Let $K$ be a field, let $r\ge3$ and $n\ge r+3$, and write
\[
 g(x)=x^r+a_{r-1}x^{r-1}+\cdots+a_1x+a_0\in K[x].
\]
If $\mathcal C_g(n)$ is MDS and of GRS type, then
\begin{equation}
 \Omega_r(g):=a_0a_2a_{r-1}^2-a_1^2a_{r-2}=0.
 \label{eq:Omega}
\end{equation}
The polynomial $\Omega_r$ is nonzero in every characteristic and has total
degree at most four.  Let $e_j(Y)$ be the $j$-th elementary symmetric
polynomial in $Y=(Y_1,\ldots,Y_r)$ and put
\begin{equation}
 \Psi_r(Y):=e_r(Y)e_{r-2}(Y)e_1(Y)^2-e_{r-1}(Y)^2e_2(Y).
 \label{eq:Psi-root}
\end{equation}
Then $\Psi_r$ is a nonzero homogeneous polynomial of degree $2r$ in every
characteristic.  Over an algebraic closure, if
$g(x)=\prod_{i=1}^r(x-\lambda_i)$, then
\begin{equation}
 \Omega_r(g)=\Psi_r(\lambda_1,\ldots,\lambda_r).
 \label{eq:Omega-root-specialization}
\end{equation}
\end{proposition}

\begin{proof}
Put
\[
 U=Q_r=(u_0,\ldots,u_{r-1})^T,
 \qquad
 V=Q_{r+1}=(v_0,\ldots,v_{r-1})^T.
\]
Since the first $r+2$ columns of $H_g^{\mathrm{rem}}$ are in linearly
general position, the MDS hypothesis gives
\begin{equation}
 u_i\ne0,\qquad v_i\ne0,\qquad
 u_iv_j-u_jv_i\ne0\quad(i\ne j).
 \label{eq:Omega-MDS-nonvanishing}
\end{equation}
For $0\le i\le r-1$, consider the binary form
\[
 X_i(s,t)=u_iv_i\prod_{j\ne i}(v_js-u_jt)
\]
of degree $r-1$.  Evaluation at $[u_i:v_i]$ annihilates $X_j$ for
$j\ne i$ and gives a nonzero value for $X_i$.  Hence the $X_i$ are linearly
independent, so they form a basis of the binary forms of degree $r-1$ and
the morphism
\[
 [s:t]\longmapsto[X_0(s,t):\cdots:X_{r-1}(s,t)]
\]
parametrizes a rational normal curve $\Gamma\subseteq\PP_K^{r-1}$.
Evaluation at $[1:0]$, $[0:1]$, and the points $[u_i:v_i]$ shows that
$\Gamma$ contains
\[
 [Q_0],\ldots,[Q_{r-1}],[U],[V].
\]
The displayed $r+2$ points are in linearly general position by the MDS
property.  Since $r+2=(r-1)+3$, the rigidity lemma
\cite[Lemma~4.3]{LiYuanCyclicOrbits}, applied after scalar extension to an
algebraic closure, shows that $\Gamma$ is the unique geometric rational
normal curve containing them.

Let $W=Q_{r+2}=(w_0,\ldots,w_{r-1})^T$.  If the code is of GRS type, then
the projective columns of the parity-check matrix $H_g^{\mathrm{rem}}$ lie
on a split rational normal curve.
The uniqueness just noted forces that curve to be $\Gamma$, so $[W]\in
\Gamma$.  If $W$ is proportional to $(X_i(s,t))_i$, the matrix
\begin{equation}
 R(W)=
 \begin{pmatrix}
  w_0v_0&-w_0u_0&-u_0v_0\\
  \vdots&\vdots&\vdots\\
  w_{r-1}v_{r-1}&-w_{r-1}u_{r-1}&-u_{r-1}v_{r-1}
 \end{pmatrix}
 \label{eq:Omega-rank-matrix}
\end{equation}
annihilates a nonzero vector.  Explicitly, after extending scalars if
necessary, write $W=c(X_i(s,t))_i$ and take
$z=c\prod_j(v_js-u_jt)$; then $R(W)(s,t,z)^T=0$.  Hence
$\operatorname{rank}R(W)\le2$, and in
particular the minor in rows $0,1,2$ vanishes.

The recurrence~\eqref{eq:remainder-recurrence} gives
\begin{align*}
 u_i&=-a_i,\\
 v_0&=a_0a_{r-1},&
 v_i&=a_ia_{r-1}-a_{i-1}\quad(1\le i\le r-1).
\end{align*}
Writing $\delta=a_{r-1}^2-a_{r-2}$, one further step gives
\[
 w_0=-a_0\delta,
 \qquad
 w_i=v_{i-1}-a_i\delta\quad(1\le i\le r-1).
\]
Direct substitution in the indicated $3\times3$ minor yields
\begin{equation}
 \det R(W)_{\{0,1,2\}}
 =a_0^3(a_0-a_1a_{r-1})
   \bigl(a_0a_2a_{r-1}^2-a_1^2a_{r-2}\bigr).
 \label{eq:Omega-minor-factorization}
\end{equation}
Here $a_0\ne0$ and $a_0-a_1a_{r-1}=-v_1\ne0$ by
\eqref{eq:Omega-MDS-nonvanishing}.  The vanishing of the minor therefore
forces~\eqref{eq:Omega}.

The two monomials in $\Omega_r$ are distinct and have coefficients $1$ and
$-1$, so they do not cancel in any characteristic; their total degrees are
at most four.  Likewise, the elementary symmetric polynomials are
algebraically independent over every prime field, and the two monomials
defining $\Psi_r$ are distinct.  Thus $\Psi_r$ is nonzero in every
characteristic; both terms have root-variable degree $2r$.  Finally, Vieta's
formulas $a_{r-j}=(-1)^je_j$ give
\eqref{eq:Omega-root-specialization}.
\end{proof}

\begin{remark}
The equation $\Omega_r=0$ is only a necessary condition for GRS type and is
not an MDS test.  No converse is asserted: in applications below the Schur
system first establishes the MDS property, and $\Omega_r\ne0$ then excludes
GRS type.  When $r=3$, this equation defines the cubic-polynomial GRS surface described in \cite[Corollary~6.9]{LiYuanCyclicOrbits}.  What is new here is the
remainder-coordinate calculation from the cited rigidity input, valid
uniformly in $r$, and its use with the rectangular Schur system to produce
completely split MDS non-GRS members.
\end{remark}

\begin{theorem}[Completely split MDS non-GRS members]
\label{thm:split-existence}
Assume $r\ge3$ and $n\ge r+3$, and put
\[
 d_{n,r}^{\mathrm{spl}}
 :=W_{n,r}+3r+\binom r2.
\]
If
\begin{equation}
 q>d_{n,r}^{\mathrm{spl}},
 \label{eq:split-existence-bound}
\end{equation}
then at least
\[
 \left\lceil
 \frac{q^{r-1}\bigl(q-d_{n,r}^{\mathrm{spl}}\bigr)}{r!}
 \right\rceil
\]
distinct monic polynomials
\[
 g(x)=\prod_{i=1}^r(x-\alpha_i),
 \qquad
 \alpha_1,\ldots,\alpha_r\in\F_q^*
 \quad\text{pairwise distinct},
\]
define MDS non-GRS codes $\mathcal C_g(n)$.
\end{theorem}

\begin{proof}
Consider the polynomial
\begin{equation}
 \mathcal F_{n,r}(Y)=
 \left(\prod_{i=1}^rY_i\right)
 \left(\prod_{i<j}(Y_j-Y_i)\right)
 \Psi_r(Y)
 \prod_{\kappa\in\Part}s_\kappa(Y).
 \label{eq:split-avoidance-polynomial}
\end{equation}
The coordinate and Vandermonde factors are nonzero, the Schur factors are
nonzero by the monomial-symmetric expansion used in the proof of
Theorem~\ref{thm:degree-primitivity}, and
Proposition~\ref{prop:Omega-obstruction} shows that $\Psi_r$ is nonzero in
every characteristic.  The degrees of the four blocks are, respectively,
$r$, $\binom r2$, $2r$, and $W_{n,r}$.  Hence
the total degree of $\mathcal F_{n,r}$ is $d_{n,r}^{\mathrm{spl}}$.  Put
\[
 \mathcal U_{n,r}(q)
 :=\{Y\in\F_q^r:\mathcal F_{n,r}(Y)\ne0\}.
\]
The finite-field zero bound gives
\[
 \#\mathcal U_{n,r}(q)
 \ge q^r-d_{n,r}^{\mathrm{spl}}q^{r-1}
 =q^{r-1}\bigl(q-d_{n,r}^{\mathrm{spl}}\bigr).
\]
For every $Y=(\alpha_1,\ldots,\alpha_r)$ in this set, the coordinate and
Vandermonde factors make the $\alpha_i$ nonzero and pairwise distinct; the
Schur block gives the MDS property by
Theorem~\ref{thm:rectangular-Schur-Pluecker}; and $\Psi_r(Y)\ne0$ gives
$\Omega_r(g)\ne0$, which certifies non-GRS type by
Proposition~\ref{prop:Omega-obstruction}.

The symmetric group $\mathfrak S_r$ acts freely on
$\mathcal U_{n,r}(q)$ by permuting the coordinates.  The coordinate,
$\Psi_r$, and Schur factors are symmetric, while the Vandermonde factor
changes only by the sign of the permutation, so $\mathcal U_{n,r}(q)$ is
$\mathfrak S_r$-stable.  Its orbits are exactly the fibers of
\[
 (\alpha_1,\ldots,\alpha_r)
 \longmapsto\prod_{i=1}^r(x-\alpha_i)
\]
on this distinct-root locus, and every such fiber has cardinality $r!$.
Consequently the number of distinct monic polynomials obtained is
\[
 \frac{\#\mathcal U_{n,r}(q)}{r!}
 \ge
 \frac{q^{r-1}\bigl(q-d_{n,r}^{\mathrm{spl}}\bigr)}{r!}.
\]
Since the left-hand side is an integer, it is at least the ceiling asserted
in the theorem.  Under~\eqref{eq:split-existence-bound} this lower bound is
positive, so existence follows as well.
\end{proof}

\section{Arithmetic strata of the coefficient MDS locus}
\label{sec:strata}

The rectangular system admits two complementary arithmetic
specializations.  Repeated roots give algebraic strata parameterized by
root centers and multiplicities.  Irreducible polynomials give Frobenius
orbits of roots over finite fields.  The same Schur coordinates control both.

\subsection{Pullback to root-multiplicity strata}

Fix a multiplicity type
\begin{equation}
 \mathbf m=(m_1,\ldots,m_s),
 \qquad m_\nu>0,
 \qquad m_1+\cdots+m_s=r.
 \label{eq:multiplicity-type}
\end{equation}
Let
\begin{equation}
 U_{\mathbf m}
 =D\!\left(
 \prod_{\nu=1}^sY_\nu
 \prod_{\mu<\nu}(Y_\nu-Y_\mu)
 \right)\subseteq\A^s_{\Z}
 \label{eq:multiplicity-parameter-space}
\end{equation}
be the ordered center space, and define
\begin{equation}
 \pi_{\mathbf m}:U_{\mathbf m}\longrightarrow\A^r_{\Z},
 \qquad
 (Y_1,\ldots,Y_s)\longmapsto
 \coeff\prod_{\nu=1}^s(x-Y_\nu)^{m_\nu}.
 \label{eq:multiplicity-map}
\end{equation}
For every partition $\kappa$, put
\begin{equation}
 S_{\kappa,\mathbf m}(Y_1,\ldots,Y_s)
 =s_\kappa(Y_1^{[m_1]},\ldots,Y_s^{[m_s]})
 \in\Z[Y_1,\ldots,Y_s].
 \label{eq:repeated-Schur}
\end{equation}
Here $Y_\nu^{[m_\nu]}$ means $m_\nu$ repetitions of the letter $Y_\nu$ in
the alphabet.  At the all-ones point,
\[
 S_{\kappa,\mathbf m}(1,\ldots,1)=s_\kappa(1^r)>0,
\]
so $S_{\kappa,\mathbf m}$ is a nonzero integral polynomial.

\begin{theorem}[Multiplicity-stratum pullback]
\label{thm:multiplicity-pullback}
As open subschemes over $\Spec\Z$, the inverse image of the coefficient MDS
locus is
the principal open subset
\begin{equation}
 \pi_{\mathbf m}^{-1}(\MDSlocus)
 =D\!\left(
 \prod_{\kappa\in\Part}S_{\kappa,\mathbf m}
 \right)\subseteq U_{\mathbf m}.
 \label{eq:multiplicity-pullback}
\end{equation}
Thus the MDS equations on a root-multiplicity stratum are exactly the
repeated-variable specializations of the rectangular Schur system.  This
equality is preserved by arbitrary base change.
\end{theorem}

\begin{proof}
On $U_{\mathbf m}$, the constant coefficient is
$(-1)^r\prod_\nu Y_\nu^{m_\nu}$ and is a unit.  The root alphabet of the
polynomial in~\eqref{eq:multiplicity-map} is
$(Y_1^{[m_1]},\ldots,Y_s^{[m_s]})$.  Substitute this alphabet into
\eqref{eq:Schur-MDS-test}.
\end{proof}

\subsection{Structural bad characteristics}

For an integral polynomial $F$, let $\cont(F)$ be the positive greatest
common divisor of its coefficients.  Define
\begin{equation}
 \mathcal B_{n,\mathbf m}
 =\left\{p\text{ prime}:
 p\mid\cont(S_{\kappa,\mathbf m})
 \text{ for some }\kappa\in\Part\right\}.
 \label{eq:bad-primes}
\end{equation}

\begin{theorem}[Structural bad characteristics]
\label{thm:structural-bad}
For a prime $p$, the following are equivalent.
\begin{enumerate}[label=\textup{(\roman*)}]
\item $p\in\mathcal B_{n,\mathbf m}$.
\item Some required repeated-variable Schur coordinate is identically zero
in characteristic $p$.
\item Over every finite extension of $\F_p$, every polynomial
\[
 g(x)=\prod_{\nu=1}^s(x-\alpha_\nu)^{m_\nu}
\]
with pairwise distinct nonzero centers fails to define an MDS code of length
$n$.
\end{enumerate}
If $p\notin\mathcal B_{n,\mathbf m}$, an MDS member of this type exists over
some finite extension of $\F_p$.
Moreover,
\[
 \mathcal B_{n,\mathbf m}
 \subseteq\{p\text{ prime}:p<n\}.
\]
\end{theorem}

\begin{proof}
The equivalence of (i) and (ii) is the definition of content.  If a required
factor is the zero polynomial modulo $p$, Theorem~\ref{thm:multiplicity-pullback}
puts the entire stratum in the non-MDS boundary over every extension.

Conversely, if $p\notin\mathcal B_{n,\mathbf m}$, every factor
$S_{\kappa,\mathbf m}$ is nonzero in $\F_p[Y_1,\ldots,Y_s]$.  Hence so is
\begin{equation}
 \left(\prod_{\nu=1}^sY_\nu\right)
 \left(\prod_{\mu<\nu}(Y_\nu-Y_\mu)\right)
 \prod_{\kappa\in\Part}S_{\kappa,\mathbf m}(Y).
 \label{eq:stratum-avoidance}
\end{equation}
A nonzero polynomial has a nonzero value over a sufficiently large finite
extension of $\F_p$.  At such a value the centers are nonzero and distinct,
and Theorem~\ref{thm:multiplicity-pullback} gives an MDS member.

It remains to prove the asserted restriction on the bad primes.  Let
$p_0\ge n$ be prime.  For every
$\kappa\subseteq(n-r)^{r-1}$, specialization of all centers to $1$ gives
\[
 S_{\kappa,\mathbf m}(1,\ldots,1)=s_\kappa(1^r).
\]
For $\kappa\ne\varnothing$, the hook--content formula
\cite[Chapter~I, \S3, Example~4]{Macdonald} reads
\[
 s_\kappa(1^r)
 =\prod_{(i,j)\in\kappa}\frac{r+j-i}{h_{ij}},
\]
where $h_{ij}$ is the hook length of the cell $(i,j)$.  Here
\[
 1\le r+j-i\le r+\kappa_1-1\le n-1<p_0,
 \qquad
 1\le h_{ij}\le\kappa_1+\ell(\kappa)-1\le n-2<p_0.
\]
Thus the numerator and denominator products both have $p_0$-adic valuation
zero, and hence $p_0\nmid s_\kappa(1^r)$.  The same conclusion is immediate
for $\kappa=\varnothing$, when the value is $1$.  Therefore no
$S_{\kappa,\mathbf m}$ is the zero polynomial modulo $p_0$, so
$p_0\notin\mathcal B_{n,\mathbf m}$.
\end{proof}

\begin{corollary}[Finite-field density on a multiplicity stratum]
\label{cor:multiplicity-field-bound}
Assume $p\notin\mathcal B_{n,\mathbf m}$, let $q$ be a power of $p$, and
put
\begin{equation}
 d_{n,r,\mathbf m}=W_{n,r}+s+\binom s2.
 \label{eq:multiplicity-degree}
\end{equation}
Then the number of ordered center tuples in $\F_q^s$ that give MDS members
is at least
\begin{equation}
 q^s-d_{n,r,\mathbf m}q^{s-1}.
 \label{eq:multiplicity-ordered-count}
\end{equation}
For $c_d=\#\{\nu:m_\nu=d\}$ define
\begin{equation}
 \operatorname{Aut}(\mathbf m)=\prod_{d\ge1}c_d!.
 \label{eq:multiplicity-aut}
\end{equation}
The number of distinct monic split polynomials of type $\mathbf m$, with
distinct nonzero centers, that give MDS codes is at least
\begin{equation}
 \max\left\{0,
 \left\lceil
 \frac{q^s-d_{n,r,\mathbf m}q^{s-1}}
      {\operatorname{Aut}(\mathbf m)}
 \right\rceil\right\}.
 \label{eq:multiplicity-polynomial-count}
\end{equation}
In particular, the sufficient existence condition is
\begin{equation}
 q>d_{n,r,\mathbf m},
 \label{eq:multiplicity-field-bound}
\end{equation}
and, for fixed $p\notin\mathcal B_{n,\mathbf m}$, the MDS density on this
split stratum over $\F_{p^e}$ is
$1-O_{n,r,\mathbf m}(p^{-e})$ as $e\to\infty$.  If moreover
$r\ge3$, $n\ge r+3$, $s\ge2$, and some $m_\nu>1$, the lower bound
in~\eqref{eq:multiplicity-polynomial-count} also counts MDS non-GRS
polynomials.
\end{corollary}

\begin{proof}
The avoidance polynomial~\eqref{eq:stratum-avoidance} is nonzero and has
degree $d_{n,r,\mathbf m}$: its Schur, coordinate, and diagonal blocks
have degrees $W_{n,r}$, $s$, and $\binom s2$, respectively.  The
finite-field zero bound proves~\eqref{eq:multiplicity-ordered-count}.  Two
ordered tuples define the same polynomial precisely when centers carrying
equal multiplicities are permuted, so every fiber has size
$\operatorname{Aut}(\mathbf m)$; this proves
\eqref{eq:multiplicity-polynomial-count}.  The full ordered stratum has
$(q-1)(q-2)\cdots(q-s)=q^s+O_s(q^{s-1})$ points, giving the density claim.
Under the additional hypotheses $r\ge3$ and $n\ge r+3$, Li and
Yuan~\cite[Corollary~5.6]{LiYuanCyclicOrbits} exclude GRS type for an MDS
polynomial having both a repeated root and a second distinct root.
\end{proof}

The inclusion of partition rectangles immediately gives, for every
multiplicity type $\mathbf m$ and every $n>r$,
\begin{equation}
 \mathcal B_{n,\mathbf m}\subseteq\mathcal B_{n+1,\mathbf m}.
 \label{eq:bad-prime-monotonicity}
\end{equation}
For a prime $p$, define
\begin{equation}
 \nu_p(\mathbf m)=
 \min\{n>r:p\in\mathcal B_{n,\mathbf m}\},
 \label{eq:nu-p-m}
\end{equation}
with value $\infty$ if the set is empty.  Then the stratum can be
structurally nondegenerate in characteristic $p$ only at lengths
$r<n<\nu_p(\mathbf m)$.

This elementary monotonicity is the multiplicity-stratum shadow of the full
length filtration in Section~\ref{sec:length}, which identifies
$\nu_p(\mathbf m)-1$ with the generic sparse-multiple threshold.

\begin{theorem}[Repeated-root Frobenius bound]
\label{thm:repeated-frobenius-bound}
Assume $s<r$, let $p$ be prime, and let $P$ be the least power of $p$
satisfying
\[
 P\ge\max_{1\le\nu\le s}m_\nu.
\]
For the generic split polynomial
\[
 g_{\mathbf m}^{\mathrm{gen}}(x)
 =\prod_{\nu=1}^s(x-Y_\nu)^{m_\nu}
 \in\F_p(Y_1,\ldots,Y_s)[x]
\]
one has
\begin{equation}
 \sigma_r(g_{\mathbf m}^{\mathrm{gen}})\le sP,
 \qquad
 \max\{r+1,p+1\}\le\nu_p(\mathbf m)\le sP+1.
 \label{eq:repeated-frobenius-bound}
\end{equation}
More precisely, there exists a partition $\kappa$ such that
\begin{equation}
 \ell(\kappa)\le r-1,\qquad
 \kappa_1=sP-r+1,\qquad
 S_{\kappa,\mathbf m}=0
 \quad\text{in }\F_p[Y_1,\ldots,Y_s].
 \label{eq:repeated-frobenius-outer-factor}
\end{equation}
Thus the pullback to the multiplicity stratum of the outer Schur layer
introduced at length $sP+1$ contains an identically vanishing factor.
More generally, every concrete split polynomial of type $\mathbf m$ in
characteristic $p$ fails to be MDS at each length $n>sP$.
\end{theorem}

\begin{proof}
In characteristic $p$,
\begin{equation}
 F_{\mathbf m}(x)=
 \prod_{\nu=1}^s(x^P-Y_\nu^P)
 \label{eq:repeated-frobenius-multiple}
\end{equation}
is divisible by $g_{\mathbf m}^{\mathrm{gen}}$.  It is a degree-$s$
polynomial in $x^P$, hence has at most $s+1\le r$ nonzero coefficients and
degree $sP$.  This proves the first inequality.

By definition, $\nu_p(\mathbf m)\ge r+1$.  The final assertion of
Theorem~\ref{thm:structural-bad} shows that
$p\notin\mathcal B_{n,\mathbf m}$ whenever $n\le p$, and therefore
$\nu_p(\mathbf m)\ge p+1$.

Moreover, $r=\sum_\nu m_\nu\le sP$ and $s+1\le r$.  Hence the exponent set
$\{0,P,2P,\ldots,sP\}$ can be enlarged to an $r$-element set
$I\subseteq\{0,1,\ldots,sP\}$ containing both $0$ and $sP$.  The
 coefficients of $F_{\mathbf m}$ give a nontrivial linear relation among the
 remainder columns indexed by $I$, with coefficient zero on the added
 indices.  The same relation annihilates the Hasse-evaluation columns, so
 $F_I=0$ in $\F_p[Y_1,\ldots,Y_s]$.  If $U=I\setminus\{0\}$, then
\[
 \kappa(U)_1=sP-(r-1)=sP-r+1,
\]
 and Proposition~\ref{prop:confluent-Schur} gives
\[
 F_I=
 \prod_{\mu<\nu}(Y_\nu-Y_\mu)^{m_\mu m_\nu}
 S_{\kappa(U),\mathbf m}.
\]
The coefficient ring $\F_p[Y_1,\ldots,Y_s]$ is an integral domain, and the
first factor is nonzero.  Hence $S_{\kappa(U),\mathbf m}=0$, proving
\eqref{eq:repeated-frobenius-outer-factor}.  At length $sP+1$ this factor
belongs to the new outer layer, proving the upper bound for
$\nu_p(\mathbf m)$ in~\eqref{eq:repeated-frobenius-bound}.  Specializing the
$Y_\nu$ in
\eqref{eq:repeated-frobenius-multiple} proves the concrete statement.
\end{proof}

The squarefree type is the opposite extreme:
\begin{equation}
 \mathcal B_{n,(1^r)}=\varnothing
 \qquad(n>r),
 \label{eq:squarefree-no-structural-bad}
\end{equation}
because the ordinary Schur polynomials in independent variables remain
nonzero in every characteristic.  The exact type $(2,1)$ formula below
shows that the Frobenius bound is sharp.  Indeed, in characteristic $p$,
\begin{equation}
 (x-\alpha)^2(x-\beta)\mid
 (x^p-\alpha^p)(x^p-\beta^p)
 =x^{2p}-(\alpha^p+\beta^p)x^p+\alpha^p\beta^p,
 \label{eq:two-one-frobenius-multiple}
\end{equation}
and the first structurally bad length is $2p+1$.

\subsection{Explicit repeated-root strata}

\begin{theorem}[Pure-power stratum]
\label{thm:pure-power-bad}
Let $r\ge2$.
\begin{enumerate}[label=\textup{(\arabic*)}]
\item If $n\ge r+2$, then
\begin{equation}
 \mathcal B_{n,(r)}=\{p\text{ prime}:p<n\}.
 \label{eq:pure-power-bad}
\end{equation}
\item If $n=r+1$, then
\begin{equation}
 \mathcal B_{r+1,(r)}
 =\left\{p:p\mid\binom rj
 \text{ for some }1\le j<r\right\}.
 \label{eq:pure-power-first-boundary}
\end{equation}
\end{enumerate}
Consequently,
\begin{equation}
 \nu_p((r))=
 \begin{cases}
 r+1,&p\mid\binom rj\text{ for some }1\le j<r,\\
 \max\{r+2,p+1\},&\text{otherwise}.
 \end{cases}
 \label{eq:pure-power-nu-p}
\end{equation}
\end{theorem}

\begin{proof}
On the type $(r)$,
$S_{\kappa,(r)}(Y)=s_\kappa(1^r)Y^{|\kappa|}$.  For $n\ge r+2$, Li and
Yuan~\cite[Lemma~5.2]{LiYuanCyclicOrbits} prove that for $\beta\ne0$ the code
attached to $(x-\beta)^r$ is MDS precisely when $n\le p$.  Combine this with
Theorem~\ref{thm:structural-bad}.  At $n=r+1$, the nonempty partitions are
the columns $(1^j)$, and
$s_{(1^j)}(1^r)=e_j(1^r)=\binom rj$.  The threshold formula follows from
the first two statements and monotonicity.
\end{proof}

\begin{theorem}[The cubic type $(2,1)$]
\label{thm:two-one-bad}
For $r=3$ and every $n>3$,
\begin{equation}
 \mathcal B_{n,(2,1)}
 =\{p\text{ prime}:2p\le n-1\},
 \qquad
 \nu_p((2,1))=2p+1.
 \label{eq:two-one-bad}
\end{equation}
\end{theorem}

\begin{proof}
Write $g(x)=(x-\alpha)^2(x-\beta)$.  For a reduced index set
$U=\{u<v\}$, direct expansion of the ordered Hasse confluent minor gives
\begin{equation}
 F_{\{0,u,v\}}(\alpha,\beta)
 =u\alpha^{u-1}\beta^v-v\alpha^{v-1}\beta^u
  +(v-u)\alpha^{u+v-1}.
 \label{eq:two-one-minor}
\end{equation}
The three displayed monomials are distinct, so the content of
$F_{\{0,u,v\}}$ is
$\gcd(u,v,v-u)=\gcd(u,v)$.  On the other hand,
Proposition~\ref{prop:confluent-Schur} gives
\[
 F_{\{0,u,v\}}=(\beta-\alpha)^2S_{\kappa(U),(2,1)}.
\]
The factor $(\beta-\alpha)^2$ is primitive.  Hence the multivariable Gauss
lemma shows that
\[
 \operatorname{cont}\bigl(S_{\kappa(U),(2,1)}\bigr)=\gcd(u,v).
\]
This confluent Schur polynomial is therefore identically zero in
characteristic $p$ exactly when $p$ divides both $u$ and $v$.  A pair
$1\le u<v\le n-1$ with this property exists exactly when $u=p,v=2p$ is
available, namely when $2p\le n-1$.  Theorem~\ref{thm:structural-bad} gives
the result.
\end{proof}

\begin{proposition}[Littlewood--Richardson expansion]
\label{prop:LR-expansion}
For every multiplicity type,
\begin{equation}
 S_{\kappa,\mathbf m}(Y)
 =\sum_{\mu^{(1)},\ldots,\mu^{(s)}}
 c^\kappa_{\mu^{(1)},\ldots,\mu^{(s)}}
 \prod_{\nu=1}^s
 s_{\mu^{(\nu)}}(1^{m_\nu})Y_\nu^{|\mu^{(\nu)}|},
 \label{eq:LR-expansion}
\end{equation}
where $c^\kappa_{\mu^{(1)},\ldots,\mu^{(s)}}$ is the iterated
Littlewood--Richardson coefficient.
\end{proposition}

\begin{proof}
Apply the Schur expansion for a disjoint union of alphabets
\cite[Chapter~I, \S5]{Macdonald} and specialize the $\nu$-th alphabet to
$m_\nu$ copies of $Y_\nu$.  Homogeneity gives the
factor $s_{\mu^{(\nu)}}(1^{m_\nu})Y_\nu^{|\mu^{(\nu)}|}$.
\end{proof}

Formula~\eqref{eq:LR-expansion} explains the structural contents through
Littlewood--Richardson multiplicities and the Weyl dimensions
$s_\mu(1^m)$.

\subsection{The irreducible Frobenius stratum}

Let $g\in\F_q[x]$ be irreducible of degree $r\ge2$, let
$L=\F_{q^r}$, and let $\alpha\in L$ be a root.  Its roots are
\begin{equation}
 \Lambda_g=(\alpha,\alpha^q,\ldots,\alpha^{q^{r-1}}).
 \label{eq:Frobenius-alphabet}
\end{equation}

\begin{proposition}[Trace and power-orbit models]
\label{prop:trace-model}
For $c=(c_0,\ldots,c_{n-1})\in\F_q^n$,
\begin{equation}
 c\in\mathcal C_g(n)
 \quad\Longleftrightarrow\quad
 \sum_{j=0}^{n-1}c_j\alpha^j=0\text{ in }L.
 \label{eq:irreducible-kernel}
\end{equation}
With respect to any $\F_q$-basis $\mathcal B$ of $L$, a parity-check
matrix is
\begin{equation}
 H_{\alpha,\mathcal B}
 =([1]_{\mathcal B}\ [\alpha]_{\mathcal B}\ \cdots
 [\alpha^{n-1}]_{\mathcal B}),
 \label{eq:power-orbit-matrix}
\end{equation}
and the dual code is
\begin{equation}
 \mathcal C_g(n)^\perp=
 \left\{
 (\Tr_{L/\F_q}(\gamma),\Tr_{L/\F_q}(\gamma\alpha),\ldots,
 \Tr_{L/\F_q}(\gamma\alpha^{n-1})):\gamma\in L
 \right\}.
 \label{eq:trace-dual}
\end{equation}
\end{proposition}

\begin{proof}
Since $g$ is the minimal polynomial of $\alpha$, it divides a polynomial
$c(x)$ exactly when $c(\alpha)=0$, proving~\eqref{eq:irreducible-kernel}.
Writing this equation in a basis gives~\eqref{eq:power-orbit-matrix}.
Every word in~\eqref{eq:trace-dual} is orthogonal to the kernel in
\eqref{eq:irreducible-kernel}.  The trace pairing is nondegenerate, and the
first $r$ powers of $\alpha$ form a basis, so both spaces have dimension
$r$ and are equal; compare Delsarte's trace description~\cite{Delsarte}.
\end{proof}

\begin{theorem}[Frobenius--Schur criterion]
\label{thm:Frobenius-Schur}
Let $g\in\F_q[x]$ be irreducible of degree $r\ge2$.  Then
\begin{equation}
 \mathcal C_g(n)\text{ is MDS}
 \Longleftrightarrow
 s_\kappa(\alpha,\alpha^q,\ldots,\alpha^{q^{r-1}})\ne0
 \quad\text{for every }\kappa\in\Part.
 \label{eq:Frobenius-Schur}
\end{equation}
Moreover, write
\begin{equation}
 \alpha^m=A_{0,m}+A_{1,m}\alpha+\cdots+A_{r-1,m}\alpha^{r-1}.
 \label{eq:power-basis-coordinates}
\end{equation}
For $U=\{u_1<\cdots<u_{r-1}\}$,
\begin{equation}
 \det(A_{i,u_j})_{1\le i,j\le r-1}
 =s_{\kappa(U)}(\alpha,\alpha^q,\ldots,\alpha^{q^{r-1}}).
 \label{eq:recurrence-Schur-bridge}
\end{equation}
The coefficients $A_{i,m}$ are computed from the recurrence
\begin{equation}
 A_{i,m+r}=-a_{r-1}A_{i,m+r-1}-\cdots-a_0A_{i,m},
 \qquad A_{i,m}=\delta_{i,m}\ (0\le m<r).
 \label{eq:Aim-recurrence}
\end{equation}
\end{theorem}

\begin{proof}
An irreducible polynomial of degree at least two has $a_0\ne0$.  Substitute
the Frobenius alphabet~\eqref{eq:Frobenius-alphabet} into
Theorem~\ref{thm:rectangular-Schur-Pluecker}.  In the power basis
$1,\alpha,\ldots,\alpha^{r-1}$, the reduced determinant containing the
column of $1$ expands along that column to the determinant on the left of
\eqref{eq:recurrence-Schur-bridge}.  The recurrence is multiplication by
$\alpha$ followed by the relation $g(\alpha)=0$.
\end{proof}

The recurrence determinant is precisely the power-basis realization of the
Frobenius--Schur coordinates.

\begin{theorem}[Large-field existence on the irreducible stratum]
\label{thm:irreducible-existence}
Fix $n>r\ge2$, and let $\mu$ be the M\"obius function.  There is an
irreducible monic polynomial $g\in\F_q[x]$ of degree $r$ for which
$\mathcal C_g(n)$ is MDS whenever
\begin{equation}
 \sum_{d\mid r}\mu(d)q^{r/d}
 >(n-1)\binom{n-1}{r-1}\frac{q^r-1}{q-1}.
 \label{eq:irreducible-exact-bound}
\end{equation}
The simpler sufficient condition
\begin{equation}
 q>(r-1)+r(n-1)\binom{n-1}{r-1}
 \label{eq:irreducible-simple-bound}
\end{equation}
also implies existence.
\end{theorem}

\begin{proof}
Let $L=\F_{q^r}$ and let
\[
 \mathcal A_r(q)=\{\alpha\in L:[\F_q(\alpha):\F_q]=r\}.
\]
The exact count is
\begin{equation}
 |\mathcal A_r(q)|=\sum_{d\mid r}\mu(d)q^{r/d}.
 \label{eq:degree-r-elements}
\end{equation}
Fix $U=\{u_1<\cdots<u_{r-1}\}\subseteq\{1,\ldots,n-1\}$.  If
 $1,\alpha^{u_1},\ldots,\alpha^{u_{r-1}}$ are dependent over $\F_q$, then
 $\alpha$ is a root of
 \[
  c_0+c_1X^{u_1}+\cdots+c_{r-1}X^{u_{r-1}}
 \]
 for some projective coefficient vector
 $[c_0:\cdots:c_{r-1}]\in\PP^{r-1}(\F_q)$.  Each such nonzero polynomial
has degree at most $n-1$, so the number of bad $\alpha$ for this fixed $U$
is at most
\begin{equation}
 (n-1)|\PP^{r-1}(\F_q)|
 =(n-1)\frac{q^r-1}{q-1}.
 \label{eq:bad-alpha-fixed-U}
\end{equation}
There are $\binom{n-1}{r-1}$ reduced sets $U$.  Under
\eqref{eq:irreducible-exact-bound}, a degree-$r$ element survives all these
bad sets.  Its minimal polynomial is irreducible, and
Proposition~\ref{prop:trace-model}, or equivalently
Theorem~\ref{thm:Frobenius-Schur}, gives the MDS property.

For the simpler estimate, every element of degree less than $r$ lies in a
proper subfield, so the union bound gives
\[
 |\mathcal A_r(q)|\ge q^r-(r-1)q^{r-1}.
\]
Also $(q^r-1)/(q-1)\le rq^{r-1}$.  Substitution shows that
\eqref{eq:irreducible-simple-bound} implies
\eqref{eq:irreducible-exact-bound}.
\end{proof}

\begin{corollary}[Density one on the irreducible stratum]
\label{cor:irreducible-count}
Let $I_q(n,r)$ be the number of monic irreducible degree-$r$ polynomials
$g$ for which $\mathcal C_g(n)$ is MDS, put
\begin{equation}
 \operatorname{Irr}_q(r)
 =\frac1r\sum_{d\mid r}\mu(d)q^{r/d},
 \qquad
 C_{n,r}=(n-1)\binom{n-1}{r-1}.
 \label{eq:Irr-and-C}
\end{equation}
Then
\begin{equation}
 \operatorname{Irr}_q(r)
 -\frac{C_{n,r}}r\frac{q^r-1}{q-1}
 \le I_q(n,r)\le\operatorname{Irr}_q(r).
 \label{eq:irreducible-count}
\end{equation}
Consequently, for fixed $n>r\ge2$ and $q\to\infty$,
\begin{equation}
 I_q(n,r)=\frac{q^r}{r}+O_{n,r}(q^{r-1}),
 \qquad
 \frac{I_q(n,r)}{\operatorname{Irr}_q(r)}\longrightarrow1.
 \label{eq:irreducible-density-one}
\end{equation}
If $r\ge3$ and $n\ge r+3$, every member counted by $I_q(n,r)$ is non-GRS,
so irreducible MDS non-GRS members likewise have density one.
\end{corollary}

\begin{proof}
There are $r\operatorname{Irr}_q(r)$ degree-$r$ elements.  The union bound
in the proof of Theorem~\ref{thm:irreducible-existence} excludes at most
$C_{n,r}(q^r-1)/(q-1)$ of them.  The good set is Frobenius-stable, and each
orbit has size $r$, proving~\eqref{eq:irreducible-count}.  The M\"obius
formula gives $\operatorname{Irr}_q(r)=q^r/r+O_r(q^{r/2})$, while the
subtracted term is $O_{n,r}(q^{r-1})$; this proves
\eqref{eq:irreducible-density-one}.  Under $r\ge3$ and $n\ge r+3$, Li and
Yuan~\cite[Corollary~5.5]{LiYuanCyclicOrbits} show that an MDS polynomial
with an irreducible factor of degree at least three is automatically non-GRS.
\end{proof}

\begin{corollary}[Growing-length density on the irreducible stratum]
\label{cor:growing-length-density}
Fix $r\ge2$, and let $n=n(q)>r$ be an integer-valued function of the prime
power $q$.  If
\[
 n(q)^r=o(q),
\]
then, as $q\to\infty$ through prime powers,
\begin{equation}
 \frac{I_q(n(q),r)}{\operatorname{Irr}_q(r)}
 =1-O_r\!\left(\frac{n(q)^r}{q}\right)\longrightarrow1.
 \label{eq:growing-irreducible-density}
\end{equation}
If $r\ge3$ and $n(q)\ge r+3$ for all sufficiently large $q$, every
irreducible MDS member counted in~\eqref{eq:growing-irreducible-density} is
non-GRS.
\end{corollary}

\begin{proof}
For fixed $r$,
\[
 C_{n,r}=(n-1)\binom{n-1}{r-1}=O_r(n^r).
\]
Divide~\eqref{eq:irreducible-count} by
$\operatorname{Irr}_q(r)$ and use
$\operatorname{Irr}_q(r)\asymp_r q^r$ together with
$(q^r-1)/(q-1)=O_r(q^{r-1})$.  The final assertion follows
from~\cite[Corollary~5.5]{LiYuanCyclicOrbits}.
\end{proof}

\subsection{Quadratic Frobenius strata as a coefficient application}
\label{subsec:quadratic-boundaries}

The Schur product tests membership in $\MDSlocus$; conditional on that test,
the rational-normal-curve classification
in~\cite[Theorem~5.3]{LiYuanCyclicOrbits} tests membership in the set of
GRS-type coefficient points over $\F_q$.  We record the resulting coefficient
criteria on two quadratic factorization strata.

\begin{proposition}[The all-quadratic factorization stratum]
\label{prop:all-quadratic-coeff}
Let $r=2v\ge4$, let $n\ge r+3$, and suppose
\[
 g(x)=\prod_{j=0}^{v-1}(x^2-\tau_jx+P_j)\in\F_q[x]
\]
is a product of pairwise distinct monic irreducible quadratics for which
$\mathcal C_g(n)$ is MDS.  Then $\mathcal C_g(n)$ is of GRS type if and only
if the factors can be reordered so that there are
$P\in\F_q^*$ and $\sigma\in\F_q$ satisfying
\begin{enumerate}[label=\textup{(\arabic*)},leftmargin=2.7em]
\item $P_0=\cdots=P_{v-1}=P$;
\item $T^2-\sigma T+1$ is irreducible over $\F_q$, and either of its roots
has multiplicative order at least $n$;
\item $\tau_{j+2}=\sigma\tau_{j+1}-\tau_j$ for $0\le j\le v-3$;
\item $\tau_{v-2}=(\sigma-1)\tau_{v-1}$;
\item
\[
 \tau_0^2+\tau_1^2-\sigma\tau_0\tau_1=(4-\sigma^2)P.
\]
\end{enumerate}
For $v=2$, condition \textup{(3)} is empty and condition \textup{(4)}
reads $\tau_0=(\sigma-1)\tau_1$.  In particular, unequal constant terms
force an MDS member of this stratum to be non-GRS.
\end{proposition}

\begin{proof}
By~\cite[Theorem~5.3, family~\textup{(N0)}]{LiYuanCyclicOrbits}, the only
possible GRS family with this factorization pattern is \textup{(N0)}.
Appendix~\ref{app:all-quadratic-reconstruction}
translates its parametrization, after reordering the Frobenius-conjugate
pairs, exactly into conditions \textup{(1)}--\textup{(5)}.
\end{proof}

\begin{proposition}[The one-linear-plus-quadratic factorization stratum]
\label{prop:mixed-quadratic-coeff}
Let $r=2v+1\ge3$, let $n\ge r+3$, and suppose
\[
 g(x)=(x-\beta)\prod_{j=1}^{v}(x^2-s_jx+P_j)\in\F_q[x],
 \qquad \beta\in\F_q^*,
\]
where the quadratic factors are pairwise distinct and irreducible and
$\mathcal C_g(n)$ is MDS.  Define
\begin{equation}
 \mathcal D_0(\sigma)=2,\qquad \mathcal D_1(\sigma)=\sigma,\qquad
 \mathcal D_{j+1}(\sigma)=\sigma\mathcal D_j(\sigma)
 -\mathcal D_{j-1}(\sigma).
 \label{eq:Dickson-recurrence}
\end{equation}
Then $\mathcal C_g(n)$ is of GRS type if and only if the quadratic factors
can be reordered so that, for some $\sigma\in\F_q$,
\begin{equation}
 P_j=\beta^2,\qquad s_j=\beta\mathcal D_j(\sigma)
 \quad(1\le j\le v),
 \label{eq:mixed-coefficient-test}
\end{equation}
and $T^2-\sigma T+1$ is irreducible over $\F_q$ with either root of
multiplicative order at least $n$.  In particular, $P_j\ne\beta^2$ for one
$j$ forces an MDS member of this stratum to be non-GRS.
\end{proposition}

\begin{proof}
By~\cite[Theorem~5.3, family~\textup{(N1)}]{LiYuanCyclicOrbits}, the only
possible GRS family with this factorization pattern is the centered nonsplit
family \textup{(N1)}.
Appendix~\ref{app:mixed-quadratic-reconstruction} translates its
parametrization into~\eqref{eq:mixed-coefficient-test} and the stated
irreducibility and order conditions.
\end{proof}

\begin{remark}[Separation of the MDS and GRS tests]
The hypotheses in Propositions~\ref{prop:all-quadratic-coeff} and
\ref{prop:mixed-quadratic-coeff} deliberately include the MDS condition:
that condition is tested by the rectangular Schur product, whereas the
recurrences test only the subsequent intersection with the GRS-type
coefficient set.
An MDS polynomial having an irreducible factor of degree at least three, or
having a repeated root together with a second distinct root, is already
automatically non-GRS when $r\ge3$ and $n\ge r+3$.
\end{remark}

\section{Length filtrations and sparse-multiple thresholds}
\label{sec:length}

\subsection{Length filtration and first-failure schemes}

The height $r-1$ of the partition rectangle is fixed by the parity-check
dimension, whereas its width $n-r$ grows with the code length.  This gives a
canonical filtration of both the MDS divisor and the sparse-multiple
problem.

\begin{theorem}[Length filtration]
\label{thm:length-filtration}
For fixed $r$ and every $n>r$,
\begin{equation}
 \mathscr D_{n+1,r}
 =\mathscr D_{n,r}
 \prod_{\substack{\ell(\kappa)\le r-1\\
                   \kappa_1=n-r+1}}
 \Schur{\kappa}.
 \label{eq:length-filtration}
\end{equation}
The same formula holds after evaluation at any monic $g$.  In particular,
$\mathscr D_{n,r}$ divides $\mathscr D_{n+1,r}$ over $\Z$.
\end{theorem}

\begin{proof}
The rectangle grows from $(n-r)^{r-1}$ to
$(n-r+1)^{r-1}$.  Its new partitions are exactly those with first part
$n-r+1$.  Apply~\eqref{eq:canonical-Schur-product}; the external factor
$A_0$ occurs on both sides.
\end{proof}

To start the filtration at the width-zero rectangle, we introduce the
following formal endpoint convention.  No coefficient code of length
$n=r$ is being introduced; throughout, the coding-theoretic construction
continues to be used only for $n>r$.  Set
\begin{equation}
 \mathscr D_{r,r}:=A_0,\qquad
 \mathcal M_{r,r}:=D(A_0),\qquad
 \Delta_{r,r}:=\operatorname{div}(A_0),
 \label{eq:length-endpoint-convention}
\end{equation}
and, for every $n\ge r$, define the outer-layer polynomial
\begin{equation}
 \mathscr O_{n,r}:=
 \prod_{\substack{\ell(\kappa)\le r-1\\
                    \kappa_1=n-r+1}}
 \Schur{\kappa}.
 \label{eq:outer-layer-polynomial}
\end{equation}
Let
\begin{equation}
 \mathcal L_{n,r}:=\operatorname{div}(\mathscr O_{n,r})
 \quad\text{on }\A_{\Z}^r.
 \label{eq:outer-layer-divisor}
\end{equation}
For $r=1$, the polynomial is the empty product $1$ and $\mathcal L_{n,1}$ is the
empty effective Cartier divisor.
The endpoint open subscheme satisfies
\[
 \mathcal M_{r,r}=D(A_0)
 \simeq\mathbb G_{m,\Z}\times\A_{\Z}^{r-1},
\]
so its structural morphism is likewise smooth and surjective.

\begin{corollary}[Polynomial and divisor filtrations]
\label{cor:outer-divisor-filtration}
For $n\ge r$ one has
\begin{align}
 \mathscr D_{n+1,r}&=\mathscr D_{n,r}\mathscr O_{n,r},
 \label{eq:outer-polynomial-filtration}\\
 \Delta_{n+1,r}&=\Delta_{n,r}+\mathcal L_{n,r}.
 \label{eq:outer-divisor-filtration}
\end{align}
Thus every new length adds exactly the effective Cartier divisor carried by
one outer Schur layer.
\end{corollary}

\begin{proof}
For $n>r$, the first identity is
Theorem~\ref{thm:length-filtration}.  For $n=r$, the width-zero rectangle
contains only the empty partition, so the same identity follows from
\eqref{eq:length-endpoint-convention} and the Schur product for length
$r+1$.  The divisor identity follows from
$\operatorname{div}(fg)=\operatorname{div}(f)+\operatorname{div}(g)$.
\end{proof}

\begin{corollary}[Enumerative data of the outer layer]
\label{cor:outer-layer-data}
For $n\ge r$ and $r\ge2$, the passage from length $n$ to length $n+1$ adds
\begin{equation}
 \binom{n-1}{r-2}
 \label{eq:outer-layer-count}
\end{equation}
Schur factors, counted with their partition labels.  Moreover,
\begin{align}
 \deg\frac{\mathscr D_{n+1,r}}{\mathscr D_{n,r}}
 &=(n-r+1)\binom{n-1}{r-2},
 \label{eq:outer-layer-degree}\\
 \rwt\!\left(\frac{\mathscr D_{n+1,r}}{\mathscr D_{n,r}}\right)
 &=\binom r2\binom{n-1}{r-1}.
 \label{eq:outer-layer-weight}
\end{align}
The size distribution of the complete rectangle is encoded by
\begin{equation}
 \sum_{\kappa\subseteq(n-r)^{r-1}}z^{|\kappa|}
 =\genfrac{[}{]}{0pt}{}{n-1}{r-1}_z.
 \label{eq:Gaussian-rectangle}
\end{equation}
Here
\begin{equation}
 \genfrac{[}{]}{0pt}{}{a}{b}_z
 :=\prod_{j=1}^{b}\frac{1-z^{a-b+j}}{1-z^j}
 \qquad(0\le b\le a)
 \label{eq:Gaussian-definition}
\end{equation}
denotes the Gaussian binomial coefficient; see
\cite[Chapter~I, \S1]{Macdonald}.
Differentiating~\eqref{eq:Gaussian-rectangle} at $z=1$ recovers the first
moment $\sum_\kappa|\kappa|=W_{n,r}$.
For $r=1$ there is no new factor and the quotient is $1$.
\end{corollary}

\begin{proof}
After fixing the first part $n-r+1$, the remaining at most $r-2$ parts
form a partition in an $(r-2)\times(n-r+1)$ rectangle, giving
\eqref{eq:outer-layer-count}.  Every new factor has ordinary degree
$n-r+1$, proving~\eqref{eq:outer-layer-degree}.  Subtracting
$W_{n,r}$ from $W_{n+1,r}$ and simplifying gives
\eqref{eq:outer-layer-weight}.  Formula~\eqref{eq:Gaussian-rectangle} is
the standard Gaussian-binomial generating function for partitions in a
rectangle.
\end{proof}

\begin{corollary}[Flatness and constant degree of an outer layer]
\label{cor:outer-layer-flatness}
For every $n\ge r\ge2$, the structural morphism
\[
 \mathcal L_{n,r}\longrightarrow\Spec\Z
\]
is flat.  Over every field $k$, its fiber is the nonzero effective Cartier
hypersurface
\[
 V\bigl((\mathscr O_{n,r})_k\bigr)\subseteq\A_k^r
\]
of exact total degree
\begin{equation}
 (n-r+1)\binom{n-1}{r-2}.
 \label{eq:outer-layer-fiber-degree}
\end{equation}
For $r=1$, one has $\mathscr O_{n,1}=1$ and $\mathcal L_{n,1}=\varnothing$, which is
flat over $\Spec\Z$.  Consequently, for $r\ge2$ the effective-divisor
identity
$\Delta_{n+1,r}=\Delta_{n,r}+\mathcal L_{n,r}$ has a nonzero Cartier
increment over $\Z$ and after base change to every field.  This is an
inequality of effective Cartier divisors; it does not by itself assert strict
containment of their supports, or strictness of the corresponding MDS opens,
in every characteristic.
\end{corollary}

\begin{proof}
The top-degree argument in the proof of
Corollary~\ref{cor:boundary-flatness} gives
\[
 \deg(\Schur{\kappa}\bmod p)=\kappa_1
\]
for every prime $p$.  Each factor in $\mathscr O_{n,r}$ therefore remains
nonzero and has exact degree $n-r+1$ in every field fiber.  Combining this
with the factor count in Corollary~\ref{cor:outer-layer-data} proves
\eqref{eq:outer-layer-fiber-degree}.  In particular, every fiber is a
nonzero effective Cartier hypersurface.
Lemma~\ref{lem:integral-hypersurface-flatness}, applied to
$\mathscr O_{n,r}$, proves flatness.  The assertion for $r=1$ is immediate.
\end{proof}

For a monic polynomial $g$ with $a_0\ne0$, define its first vanishing Schur
width by
\begin{equation}
 \tau(g)
 :=\min\{\kappa_1:\ell(\kappa)\le r-1,
                    \ \Schur{\kappa}(g)=0\},
 \label{eq:tau-definition}
\end{equation}
with value $\infty$ when the displayed set is empty.

\begin{lemma}[Invariance under scalar extension]
\label{lem:scalar-extension-thresholds}
Let $L/K$ be a field extension, and let $g\in K[x]$ be monic of degree $r$
with $g(0)\ne0$.  If the subscripts indicate the coefficient field, then
\begin{equation}
 \sigma_{r,K}(g)=\sigma_{r,L}(g),
 \qquad
 \tau_K(g)=\tau_L(g).
 \label{eq:scalar-extension-thresholds}
\end{equation}
\end{lemma}

\begin{proof}
For every partition $\kappa$, the value $\Schur{\kappa}(g)$ belongs to $K$,
and it vanishes in $K$ if and only if its image vanishes in $L$.  This
proves the assertion for $\tau$.

For every integer $d$, the inequality $\sigma_{r,K}(g)\le d$ is equivalent
to the existence of a linearly dependent subset of at most $r$ columns
among
\[
 Q_0,Q_1,\ldots,Q_d.
\]
These columns lie in $K^r$, and the rank of every such column submatrix is
unchanged after scalar extension from $K$ to $L$.  Hence the same integers
$d$ satisfy $\sigma_{r,K}(g)\le d$ and $\sigma_{r,L}(g)\le d$, proving the
equality of the minima.
\end{proof}

\begin{theorem}[First Schur zero and first sparse multiple]
\label{thm:sigma-tau}
For every monic $g$ with $a_0\ne0$,
\begin{equation}
 \boxed{\ \sigma_r(g)=r+\tau(g)-1\ }.
 \label{eq:sigma-tau}
\end{equation}
More precisely,
\begin{equation}
 \mathcal C_g(n)\text{ is MDS}
 \quad\Longleftrightarrow\quad
 r<n\le\sigma_r(g).
 \label{eq:MDS-length-interval}
\end{equation}
Thus the first vanishing Schur width determines, through
\eqref{eq:sigma-tau}, the degree of the first coefficient-weight-$\le r$
multiple.
\end{theorem}

\begin{proof}
Proposition~\ref{prop:remainder-parity-sparse} gives
\[
 \mathcal C_g(n)\text{ MDS}\Longleftrightarrow n\le\sigma_r(g).
\]
The Schur criterion gives the same property exactly when the rectangle of
width $n-r$ contains no vanishing factor, namely when
$n-r<\tau(g)$.  These inequalities describe the same initial interval of
integer lengths, so their endpoints agree and yield~\eqref{eq:sigma-tau}.
\end{proof}

\begin{corollary}[First-failure strata]
\label{cor:first-failure-strata}
Let $r\ge2$.  For every $n\ge r$, define the first-failure subscheme
\begin{equation}
 \mathcal X_{n,r}:=
 V_{\mathcal M_{n,r}}(\mathscr O_{n,r})
 =\mathcal M_{n,r}\times_{\A^r_{\Z}}V(\mathscr O_{n,r}).
 \label{eq:first-failure-scheme}
\end{equation}
Its underlying topological space is
\begin{equation}
 |\mathcal X_{n,r}|
 =|\mathcal M_{n,r}|\setminus|\mathcal M_{n+1,r}|.
 \label{eq:first-failure-underlying-set}
\end{equation}
For every field $K$, on the coefficient space $D(A_0)(K)$ one has
\begin{align}
 \{g:\sigma_r(g)\ge n\}&=\mathcal M_{n,r}(K),
 \label{eq:sigma-superlevel}\\
 \{g:\sigma_r(g)=n\}&=\mathcal X_{n,r}(K).
 \label{eq:sigma-exact-level}
\end{align}
More intrinsically, evaluate $\sigma_r$ at a point of $D(A_0)$ over its
residue field.  Then
\[
 g\longmapsto\sigma_r(g)\in\{r,r+1,\ldots,\infty\}
\]
is Zariski lower semicontinuous: every superlevel set is open.  Each finite
level set is locally closed, and set-theoretically
\begin{equation}
 D(A_0)=\coprod_{n\ge r}\{\sigma_r=n\}
 \ \coprod\ \{\sigma_r=\infty\}.
 \label{eq:sigma-stratification}
\end{equation}
\end{corollary}

\begin{proof}
For $n>r$, Theorem~\ref{thm:sigma-tau} shows that $\sigma_r(g)\ge n$ is
equivalent to the MDS condition at length $n$.  At the formal endpoint
$n=r$, the inequality $\sigma_r(g)\ge r$ is automatic, while
$\mathcal M_{r,r}=D(A_0)$ by~\eqref{eq:length-endpoint-convention}; hence
the same superlevel identity holds by convention.  Moreover,
\[
 \mathcal M_{n+1,r}
 =D_{\mathcal M_{n,r}}(\mathscr O_{n,r})
 \]
by~\eqref{eq:outer-polynomial-filtration}.  Thus $\sigma_r(g)=n$ means
membership in $\mathcal M_{n,r}$ and failure at length $n+1$, which is
exactly~\eqref{eq:sigma-exact-level}.  The same principal-open identity and
\eqref{eq:first-failure-scheme} give
\eqref{eq:first-failure-underlying-set}.  The superlevel sets are therefore
open and the finite level sets locally closed; the disjoint decomposition is
immediate.
\end{proof}

\begin{corollary}[Geometry of the first-failure subschemes]
\label{cor:first-failure-geometry}
Let $N\ge r\ge2$.  The closed immersion
\[
 \mathcal X_{N,r}\hookrightarrow\mathcal M_{N,r}
\]
is an effective Cartier divisor.  Its structural morphism
\[
 \mathcal X_{N,r}\longrightarrow\Spec\Z
\]
 is flat and is a local complete intersection morphism.  Every nonempty field
 fiber is a pure $(r-1)$-dimensional local complete intersection scheme.
\end{corollary}

\begin{proof}
Let
$j:\mathcal M_{N,r}\hookrightarrow\A^r_{\Z}$ be the open immersion.  The
pullback along $j$ of the effective Cartier divisor
$\mathcal L_{N,r}=V(\mathscr O_{N,r})$ is the effective Cartier divisor
$\mathcal X_{N,r}$ on $\mathcal M_{N,r}$.  Moreover,
$\mathcal X_{N,r}\to\mathcal L_{N,r}$ is the base change of $j$, so
$\mathcal X_{N,r}$ is an open subscheme of $\mathcal L_{N,r}$.  Since
$\mathcal L_{N,r}\to\Spec\Z$ is flat by
Corollary~\ref{cor:outer-layer-flatness}, its restriction to
$\mathcal X_{N,r}$ is flat.  As an effective Cartier divisor in the smooth
 $\Z$-scheme $\mathcal M_{N,r}$, its structural morphism is a local complete
 intersection morphism.

For a field $K$, the scheme $\mathcal M_{N,r,K}$ is a smooth geometrically
integral open subscheme of $\A_K^r$ of pure dimension $r$.  If
$\mathcal X_{N,r,K}$ is nonempty, it is an effective Cartier divisor in
this regular scheme, and hence is a pure $(r-1)$-dimensional local complete
intersection.
\end{proof}

\begin{corollary}[Generic structural and sparse thresholds]
\label{cor:generic-structural-sparse}
For every multiplicity type $\mathbf m$ and prime $p$, let
\[
 g_{\mathbf m}^{\mathrm{gen}}(x)
 =\prod_{\nu=1}^s(x-Y_\nu)^{m_\nu}
 \in\F_p(Y_1,\ldots,Y_s)[x].
\]
Then
\begin{equation}
 \boxed{\ \nu_p(\mathbf m)
 =\sigma_r(g_{\mathbf m}^{\mathrm{gen}})+1\ },
 \qquad
 \tau(g_{\mathbf m}^{\mathrm{gen}})=\nu_p(\mathbf m)-r,
 \label{eq:generic-structural-sparse}
\end{equation}
with the convention that these identities preserve $\infty$.
\end{corollary}

\begin{proof}
By Theorems~\ref{thm:multiplicity-pullback} and
\ref{thm:structural-bad}, the generic polynomial fails the MDS test at
length $n$ exactly when $p\in\mathcal B_{n,\mathbf m}$.  Thus its first
failing length is $\nu_p(\mathbf m)$.  Equation~\eqref{eq:MDS-length-interval}
identifies the same length with
$\sigma_r(g_{\mathbf m}^{\mathrm{gen}})+1$, and
Theorem~\ref{thm:sigma-tau} gives the formula for $\tau$.
\end{proof}

\begin{corollary}[Dilation and reciprocal invariance of the thresholds]
\label{cor:symmetry-thresholds}
For every monic $g\in K[x]$ with nonzero constant coefficient and every
$c\in K^*$,
\begin{align}
 \sigma_r(g_c)&=\sigma_r(g),&
 \tau(g_c)&=\tau(g),
 \label{eq:dilation-thresholds}\\
 \sigma_r(g^\#)&=\sigma_r(g),&
 \tau(g^\#)&=\tau(g).
 \label{eq:reciprocal-thresholds}
\end{align}
Consequently, for $N\ge r\ge2$, the first-failure subscheme
$\mathcal X_{N,r}$ is invariant under the $\Gm\rtimes C_2$-action on
$D(A_0)$.
\end{corollary}

\begin{proof}
Propositions~\ref{prop:scaling-action} and
\ref{prop:reciprocal-symmetry}, applied at every length, show that $g$, $g_c$,
and $g^\#$ have the same MDS length interval.  Use
Theorem~\ref{thm:sigma-tau}.  Scheme-theoretic invariance follows because
the scaling action multiplies every outer-layer factor by a unit, while
reciprocity permutes the factors in a fixed outer layer up to units; see
\eqref{eq:scaling-action} and~\eqref{eq:reciprocal-Schur}.
\end{proof}

\begin{corollary}[Scaling quotient of a first-failure layer]
\label{cor:first-failure-scaling-quotient}
Let $N>r\ge2$, and put
\[
 \mathcal X_{N,r}^{(1)}
 :=\mathcal X_{N,r}\cap V(A_{r-1}-1).
\]
Then the product decomposition of
Proposition~\ref{prop:scaling-quotient} restricts to
\begin{equation}
 \mathcal X_{N,r}
 \simeq\Gm\times\mathcal X_{N,r}^{(1)}.
 \label{eq:first-failure-scaling-product}
\end{equation}
Hence, for every finite field $\F_q$,
\begin{equation}
 \#\mathcal X_{N,r}(\F_q)
 =(q-1)\#\mathcal X_{N,r}^{(1)}(\F_q).
 \label{eq:first-failure-count-divisibility}
\end{equation}
If a field fiber of $\mathcal X_{N,r}$ is nonempty, the corresponding
fiber of $\mathcal X_{N,r}^{(1)}$ is pure of dimension $r-2$.
\end{corollary}

\begin{proof}
The outer-layer polynomial $\mathscr O_{N,r}$ is root-weight homogeneous.
Under Proposition~\ref{prop:scaling-quotient}, its pullback is therefore a
unit power of the first coordinate times its restriction to
$A_{r-1}=1$.  Its zero scheme has the product
decomposition~\eqref{eq:first-failure-scaling-product}.  The count and
dimension statements follow from
Corollary~\ref{cor:first-failure-geometry}.
\end{proof}

\begin{remark}
Theorem~\ref{thm:sigma-tau} identifies four versions of one threshold: the
first vanishing normalized Pl\"ucker coordinate, the first vanishing Schur
factor, the first dependent $r$-subset in the companion orbit, and the first
sparse multiple of $g$ with at most $r$ nonzero coefficients.
\end{remark}

\subsection{Exact and effective thresholds}

Li and Yuan~\cite[Lemma~3.5]{LiYuanCyclicOrbits} prove the MDS implication
for lengths $n\le\ord(t)$.  The next theorem strengthens that result by
determining the exact first-failure and sparse-multiple thresholds, and by
recording the automatic descent $\gamma^N\in K^*$ in the finite-order case.

\begin{theorem}[Geometric-progression threshold]
\label{thm:semisimple-threshold}
Let $K$ be a field, let $g\in K[x]$ be monic of degree $r\ge2$, and let
$L/K$ be a splitting field of $g$.  Suppose that the roots of $g$ in $L$
are pairwise distinct and have the form
\begin{equation}
 \{\gamma,\gamma t,\ldots,\gamma t^{r-1}\},
 \qquad \gamma,t\in L^*.
 \label{eq:progression-roots}
\end{equation}
If $t$ has infinite multiplicative order, then
\begin{equation}
 \sigma_r(g)=\tau(g)=\infty.
 \label{eq:infinite-progression-threshold}
\end{equation}
If $t$ has finite order $N=\ord(t)$, then $\gamma^N\in K^*$ automatically,
necessarily $N\ge r$, and
\begin{equation}
 \sigma_r(g)=N,
 \qquad
 \tau(g)=N-r+1.
 \label{eq:semisimple-threshold}
\end{equation}
\end{theorem}

\begin{proof}
Suppose first that $t$ has finite order $N$.  Distinctness of the displayed
roots gives $N\ge r$.  Let $R_N\in K[x]$ be the remainder of $x^N$ modulo
$g$, so $\deg R_N<r$.  In the splitting field $L$, for $0\le j<r$ one has
\[
 R_N(\gamma t^j)=(\gamma t^j)^N=\gamma^N.
\]
Thus $R_N-\gamma^N\in L[x]$ has the $r$ distinct roots
$\gamma,\gamma t,\ldots,\gamma t^{r-1}$ and degree less than $r$.
Consequently $R_N=\gamma^N$.  Since $R_N$ has coefficients in $K$, this
also proves $\gamma^N\in K^*$ and
\begin{equation}
 g(x)\mid x^N-\gamma^N\qquad\text{in }K[x].
 \label{eq:automatic-progression-binomial}
\end{equation}

We now treat both possible orders of $t$ simultaneously.
By Lemma~\ref{lem:scalar-extension-thresholds}, scalar extension from $K$ to
the splitting field $L$ does not change either threshold.  We may therefore
compute over $L$.  The root-evaluation parity-check matrix has entries
\[
 (\gamma t^j)^i=\gamma^i(t^i)^j,
 \qquad 0\le j<r,\quad 0\le i<n.
\]
Every maximal minor is a nonzero column-scaling factor times a Vandermonde
determinant in the parameters $t^i$.  Hence all maximal minors are nonzero
exactly while the relevant exponents are distinct modulo the order of $t$.
If $t$ has infinite order, this holds at every length and gives
\eqref{eq:infinite-progression-threshold}.  If $\ord(t)=N$, it holds for
every $n\le N$.  Equation~\eqref{eq:automatic-progression-binomial} gives
$\sigma_r(g)\le N$.  If $N=r$, the reverse inequality follows from
$\deg g=r$, while if $N>r$, MDS at length $N$ gives it
through~\eqref{eq:MDS-length-interval}.
The formula for $\tau$ follows from Theorem~\ref{thm:sigma-tau}.
\end{proof}

\begin{corollary}[Cyclotomic nonemptiness of first-failure strata]
\label{cor:nonempty-first-failure}
Fix $r\ge2$ and $N\ge r$.  If a field $K$ contains a primitive $N$-th root
of unity, then $\mathcal X_{N,r}(K)$ is nonempty.  Consequently, the
$\mathbb Q$-scheme $(\mathcal X_{N,r})_{\mathbb Q}$ is geometrically
nonempty; indeed, it has a $\mathbb Q(\zeta_N)$-rational point.  In
particular,
\begin{equation}
 (\mathcal M_{N+1,r})_{\mathbb Q}
 \subsetneq
 (\mathcal M_{N,r})_{\mathbb Q}.
 \label{eq:strict-length-filtration}
\end{equation}
Thus $\mathcal X_{N,r,\mathbb Q}$ is geometrically nonempty for every finite
threshold $N\ge r$, and the MDS-open filtration is strict over $\mathbb Q$.
Moreover,
$\mathcal X_{N,r}(\F_q)$ is nonempty whenever $N\mid(q-1)$.
\end{corollary}

\begin{proof}
Let $\zeta\in K$ be a primitive $N$-th root of unity and set
\[
 g_{N,r}(x):=\prod_{j=0}^{r-1}(x-\zeta^j).
\]
Since $N\ge r$, its roots are pairwise distinct and its constant coefficient
is nonzero.  Theorem~\ref{thm:semisimple-threshold}, with
$\gamma=1$ and $t=\zeta$, gives $\sigma_r(g_{N,r})=N$.  Thus
$g_{N,r}\in\mathcal X_{N,r}(K)$ by
Corollary~\ref{cor:first-failure-strata}.  Taking
$K=\mathbb Q(\zeta_N)$ produces a rational point of
$(\mathcal X_{N,r})_{\mathbb Q}$ over this cyclotomic extension, so the
$\mathbb Q$-scheme is geometrically nonempty.  If the two $\mathbb Q$-schemes
in~\eqref{eq:strict-length-filtration} were equal, their base changes to
$\mathbb Q(\zeta_N)$ would be equal, contradicting the constructed point in
their difference.  Finally, take $K=\F_q$ when $N\mid(q-1)$.
\end{proof}

\begin{corollary}[Faithful flatness away from the cyclotomic primes]
\label{cor:first-failure-faithfully-flat}
Let $N\ge r\ge2$, and put $S_N=\Spec\Z[1/N]$.  Then
\begin{equation}
 (\mathcal X_{N,r})_{S_N}\longrightarrow S_N
 \label{eq:first-failure-faithfully-flat}
\end{equation}
is a faithfully flat local complete intersection morphism of finite
presentation, and every fiber is pure of dimension $r-1$.  Consequently, if
$K$ is a field of characteristic zero or of characteristic $\ell\nmid N$,
then
\begin{equation}
 (\mathcal M_{N+1,r})_K
 \subsetneq
 (\mathcal M_{N,r})_K
 \label{eq:strict-filtration-away-from-N}
\end{equation}
is a strict open immersion and remains strict after every field extension of
$K$; in particular, it is geometrically strict.
\end{corollary}

\begin{proof}
Flatness and the local complete intersection property follow from
Corollary~\ref{cor:first-failure-geometry} after restriction to $S_N$.  Let
$\ell\nmid N$ be a prime.  The class of $\ell$ has finite order in
$(\Z/N\Z)^*$, so $N\mid(\ell^e-1)$ for some $e\ge1$.  Hence
$\F_{\ell^e}$ contains a primitive $N$-th root of unity, and
Corollary~\ref{cor:nonempty-first-failure} gives
\[
 \mathcal X_{N,r}(\F_{\ell^e})\ne\varnothing.
\]
Thus the fiber over $(\ell)$ is geometrically nonempty.  The generic fiber
is geometrically nonempty because it has a $\mathbb Q(\zeta_N)$-rational
point.  Every fiber over $S_N$ is therefore nonempty, so the flat morphism
\eqref{eq:first-failure-faithfully-flat} is surjective and hence faithfully
flat.  Corollary~\ref{cor:first-failure-geometry} also shows that all its
fibers are pure $(r-1)$-dimensional local complete intersection schemes.

For the stated fields $K$, the morphism $\Spec K\to S_N$ is defined.
Faithful flatness is preserved by base change, so $\mathcal X_{N,r,K}$ is
nonempty and remains nonempty after every extension of $K$.  Its underlying
space is the complement of $(\mathcal M_{N+1,r})_K$ in
$(\mathcal M_{N,r})_K$ by~\eqref{eq:first-failure-underlying-set}.  This
proves~\eqref{eq:strict-filtration-away-from-N} and its geometric strictness.
\end{proof}

\begin{remark}[Location of possible empty special fibers]
\label{rem:empty-fibers-divide-N}
Corollary~\ref{cor:first-failure-faithfully-flat} shows that the image of
$\mathcal X_{N,r}\to\Spec\Z$ contains $\Spec\Z[1/N]$.  Consequently, every
possible empty special fiber lies over a prime divisor of $N$.  This does not
assert that every fiber over a prime divisor of $N$ is empty.

The restriction on the characteristic is necessary.  For example, let
$r=2$, $N=4$, and let $K$ have characteristic $2$.  The only new partition
in the outer layer is $(3)$, and in two variables
\[
 \mathscr O_{4,2}=\Schur{(3)}=h_3=e_1^3-2e_1e_2=e_1^3.
\]
On $\mathcal M_{4,2,K}$ the earlier Schur factor
$\Schur{(1)}=e_1$ is already invertible.  Hence
\[
 \mathcal X_{4,2,K}
 =V_{\mathcal M_{4,2,K}}(\mathscr O_{4,2})
 =\varnothing.
\]
\end{remark}

\begin{corollary}[Pure-power sparse threshold]
\label{thm:pure-power-threshold}
Let $r\ge2$, let $\beta\in K^*$, and put $g(x)=(x-\beta)^r$.  Then
\begin{equation}
 \sigma_r(g)=
 \begin{cases}
 \infty,&\operatorname{char}K=0,\\
 p,&\operatorname{char}K=p>0\text{ and }r<p,\\
 r,&\operatorname{char}K=p>0\text{ and }
       \binom rj\equiv0\pmod p
       \text{ for some }1\le j<r,\\
 r+1,&\operatorname{char}K=p>0,\ r\ge p,\text{ and }
       \binom rj\not\equiv0\pmod p
       \text{ for all }1\le j<r.
 \end{cases}
 \label{eq:pure-power-sigma}
\end{equation}
In every case $\tau(g)=\sigma_r(g)-r+1$.
\end{corollary}

\begin{proof}
In characteristic zero,
$\Schur{\kappa}(g)=s_\kappa(1^r)\beta^{|\kappa|}\ne0$ for every partition
$\kappa$ of length at most $r$.  Hence the code is MDS at every length and
$\sigma_r(g)=\infty$.

If $\operatorname{char}K=p>0$, then for every partition $\kappa$ one has
\[
 \Schur{\kappa}(g)
 =S_{\kappa,(r)}(\beta)
 =s_\kappa(\beta,\ldots,\beta)
 =\beta^{|\kappa|}s_\kappa(1^r).
\]
Since $\beta\ne0$, the concrete point $g=(x-\beta)^r$ has exactly the same
Schur-coordinate vanishing pattern as the pure-power stratum.
Theorem~\ref{thm:pure-power-bad} therefore shows that its first failing
length is $\nu_p((r))$.  Hence~\eqref{eq:MDS-length-interval} gives
\[
 \sigma_r(g)=\nu_p((r))-1.
\]
Substitution of~\eqref{eq:pure-power-nu-p} yields the three
positive-characteristic cases in~\eqref{eq:pure-power-sigma}.  The formula
for $\tau$ follows from Theorem~\ref{thm:sigma-tau}.
\end{proof}

For completeness, if $r=1$ and $g=x-\beta$ with $\beta\ne0$, then
$\sigma_1(g)=\tau(g)=\infty$ in every characteristic.

\begin{corollary}[The GRS part of a first-failure stratum]
\label{cor:GRS-first-failure}
Let $r\ge3$, let $N\ge r+3$, and let $q$ be a prime power of
characteristic $p$.  Define
\begin{equation}
 \mathcal X^{\mathrm{GRS}}_{N,r}(q)
 :=\mathcal X_{N,r}(\F_q)\cap\mathcal G_{N,r}(\F_q)
 =\{g\in\mathcal G_{N,r}(\F_q):\sigma_r(g)=N\}.
 \label{eq:GRS-first-failure-definition}
\end{equation}
Then
\begin{equation}
 \mathcal X^{\mathrm{GRS}}_{N,r}(q)
 =\mathcal G_{N,r}(\F_q)\setminus
  \mathcal G_{N+1,r}(\F_q),
 \label{eq:GRS-first-failure-difference}
\end{equation}
and
\begin{equation}
 \boxed{
 \#\mathcal X^{\mathrm{GRS}}_{N,r}(q)
 =
 \begin{cases}
  q-1,
    &N=p,\\[1mm]
  \dfrac{\varphi(N)}2(q-1),
    &N\mid(q-1)\text{ or }N\mid(q+1),\\[2mm]
  0,&\text{otherwise}.
 \end{cases}}
 \label{eq:GRS-first-failure-count}
\end{equation}
The three cases in~\eqref{eq:GRS-first-failure-count} are mutually
exclusive.  Thus the GRS part of the stratum is nonempty exactly when
$N=p$, $N\mid(q-1)$, or $N\mid(q+1)$.  Moreover,
\begin{equation}
 \#\mathcal X^{\mathrm{GRS}}_{N,r}(q)
 \le \frac{\varphi(N)}2(q-1).
 \label{eq:GRS-first-failure-linear-bound}
\end{equation}
\end{corollary}

\begin{proof}
Under the stated hypotheses $r\ge3$ and $N\ge r+3$, the classification
underlying
\eqref{eq:exact-GRS-count} separates the GRS members into the unipotent and
the split and nonsplit semisimple families.  In the unipotent family,
$g=(x-\beta)^r$ with $\beta\in\F_q^*$, and the length condition is
$n\le p$.  Since the member is present at length $N\ge r+3$, this forces
$p\ge N>r$, and
    Corollary~\ref{thm:pure-power-threshold} gives $\sigma_r(g)=p$.

In either semisimple family the roots form a geometric progression with
ratio $t$ of order $e$, where $e\mid(q-1)$ in the split case and
$e\mid(q+1)$ in the nonsplit case.  The GRS length condition is $n\le e$.
Since $g\in\F_q[x]$, Theorem~\ref{thm:semisimple-threshold} applies directly
and gives $\sigma_r(g)=e$.  Consequently, a GRS member at length $N$ remains a GRS
member at length $N+1$ exactly when its sparse threshold is at least
$N+1$.  This proves~\eqref{eq:GRS-first-failure-difference}.

Now subtract the exact counts~\eqref{eq:exact-GRS-count} at lengths $N$ and
$N+1$.  The unipotent contribution is
\[
 (q-1)\bigl(\mathbf 1_{\{N\le p\}}
 -\mathbf 1_{\{N+1\le p\}}\bigr)
 =(q-1)\mathbf 1_{\{N=p\}},
\]
whereas, for a cyclic group of order $M$,
\[
 \mathcal E_M(N)-\mathcal E_M(N+1)
 =\varphi(N)\mathbf 1_{\{N\mid M\}}.
\]
Taking $M=q-1$ and $M=q+1$ gives
the two possible semisimple contributions.  Since $N\ge r+3\ge6$ and
$\gcd(q-1,q+1)\le2$, the integer $N$ cannot divide both $q-1$ and $q+1$.
If $N=p$, then $p\nmid(q-1)(q+1)$, so the unipotent contribution cannot
coincide with a semisimple one.  This proves the mutually exclusive formula
\eqref{eq:GRS-first-failure-count}.  Finally, $\varphi(N)\ge2$ for $N\ge3$,
and~\eqref{eq:GRS-first-failure-linear-bound} follows.
\end{proof}

\begin{corollary}[Generic non-GRS behavior on fixed first-failure strata]
\label{cor:generic-non-GRS-first-failure}
Let $r\ge3$, let $N\ge r+3$, and let $K$ be an algebraically closed field.
Assume that $\mathcal X_{N,r,K}\ne\varnothing$, and let
\[
 \mathcal Z^{\mathrm{GRS}}_{N,r,K}
 :=\overline{\{a\in\mathcal X_{N,r,K}(K):
       \mathcal C_{g_a}(N)\text{ is of GRS type}\}}
       ^{\,\mathcal X_{N,r,K}},
\]
where $g_a$ is the monic polynomial with coefficient point $a$ and the
closure is endowed with its reduced induced structure.  Then
\begin{equation}
 \dim \mathcal Z^{\mathrm{GRS}}_{N,r,K}\le1.
 \label{eq:fixed-layer-GRS-dimension}
\end{equation}
Consequently, for every irreducible component $C$ of
$\mathcal X_{N,r,K}$, the set
\[
 C\setminus\mathcal Z^{\mathrm{GRS}}_{N,r,K}
\]
is a nonempty dense open subset of $C$, and, whenever the intersection is
nonempty,
\begin{equation}
 \operatorname{codim}_C
 \bigl(C\cap\mathcal Z^{\mathrm{GRS}}_{N,r,K}\bigr)\ge r-2.
 \label{eq:generic-non-GRS-codimension}
\end{equation}
Every coefficient point $a$ in this dense open subset corresponds to the
monic polynomial $g_a$ satisfying
\[
 \sigma_r(g_a)=N,
 \qquad
 \mathcal C_{g_a}(N)\text{ is MDS and non-GRS}.
\]

The nonemptiness hypothesis is automatic when $\operatorname{char}K=0$, or
when $\operatorname{char}K=\ell\nmid N$.  Consequently, for every prime
$\ell\nmid N$, there exist $e\ge1$ and a monic polynomial
$g\in\F_{\ell^e}[x]$ of degree $r$ such that $\sigma_r(g)=N$, and
$\mathcal C_g(N)$ is MDS and remains non-GRS after scalar extension to
$\overline{\F}_\ell$.
\end{corollary}

\begin{proof}
By Corollary~\ref{cor:first-failure-geometry}, every irreducible component
of the nonempty scheme $\mathcal X_{N,r,K}$ has dimension $r-1$.

Let
\[
 T_N(K):=\{t\in K^*: \ord(t)=N\}.
\]
This is a finite set.  For $t\in T_N(K)$, consider the coefficient morphism
\[
 \phi_t:\mathbb G_{m,K}\longrightarrow\A_K^r,
 \qquad
 \gamma\longmapsto
 \coeff\!\left(\prod_{j=0}^{r-1}(x-\gamma t^j)\right),
\]
and put
\[
 W_{N,r,K}^{\mathrm{ss}}
 :=\bigcup_{t\in T_N(K)}
   \overline{\phi_t(\mathbb G_{m,K})}^{\,\A_K^r}.
\]
This is a finite union of closed subsets of dimension at most one.  Also put
\[
 \psi:\mathbb G_{m,K}\longrightarrow\A_K^r,
 \qquad
 \beta\longmapsto\coeff((x-\beta)^r),
 \qquad
 W_{r,K}^{\mathrm{pp}}
 :=\overline{\psi(\mathbb G_{m,K})}^{\,\A_K^r}.
\]

Under the stated bounds $r\ge3$ and $N\ge r+3$, Li and
Yuan~\cite[Theorem~6.6]{LiYuanCyclicOrbits} show that every GRS point
of $\mathcal X_{N,r,K}(K)$ is either semisimple, with pairwise distinct roots
\[
 \{\gamma,\gamma t,\ldots,\gamma t^{r-1}\},
\]
or belongs to the pure-power family.  In the semisimple case,
Theorem~\ref{thm:semisimple-threshold} gives
\[
 \sigma_r(g)=\ord(t),
\]
with the value $\infty$ when $t$ has infinite order.  Since a point of
$\mathcal X_{N,r,K}$ has $\sigma_r(g)=N$, one must have $t\in T_N(K)$.
The pure-power points are contained in $W_{r,K}^{\mathrm{pp}}$.  Hence all
GRS $K$-points of the fixed layer are contained in
\[
 \mathcal X_{N,r,K}\cap
 \bigl(W_{N,r,K}^{\mathrm{ss}}\cup W_{r,K}^{\mathrm{pp}}\bigr).
\]
The set on the right is closed in $\mathcal X_{N,r,K}$ and has dimension at
most one.  It therefore contains $\mathcal Z^{\mathrm{GRS}}_{N,r,K}$,
proving~\eqref{eq:fixed-layer-GRS-dimension}.

Since $r-1>1$, no irreducible component of $\mathcal X_{N,r,K}$ is
contained in $\mathcal Z^{\mathrm{GRS}}_{N,r,K}$.  The density and
codimension assertions follow.  Corollary~\ref{cor:first-failure-strata}
shows that every point of the complement has $\sigma_r(g)=N$ and is MDS at
length $N$; it is non-GRS by the definition of
$\mathcal Z^{\mathrm{GRS}}_{N,r,K}$.

The stated nonemptiness follows from
Corollary~\ref{cor:first-failure-faithfully-flat}.  Finally take
$K=\overline{\F}_\ell$.  The nonempty open complement of
$\mathcal Z^{\mathrm{GRS}}_{N,r,K}$ has a $K$-point, whose finitely many
coefficients lie in some $\F_{\ell^e}$.  Lemma~\ref{lem:scalar-extension-thresholds}
shows that the resulting polynomial over $\F_{\ell^e}$ still satisfies
$\sigma_r(g)=N$.  Viewed again over $K$, its coefficient point is the chosen
point outside $\mathcal Z^{\mathrm{GRS}}_{N,r,K}$; hence
$\mathcal C_g(N)\otimes_{\F_{\ell^e}}K$ is non-GRS.
\end{proof}

\begin{proposition}[Companion projective period]
\label{prop:general-projective-order}
Let $g\in\F_q[x]$ be monic of degree $r\ge2$ with $a_0\ne0$, let $T_g$ be
its companion multiplication matrix, and set
\begin{equation}
 N_g:=\ord_{\operatorname{PGL}_r(\F_q)}([T_g]).
 \label{eq:Ng-definition}
\end{equation}
Then the projective remainder orbit has exact period $N_g$:
\begin{equation}
 [Q_i]=[Q_j]
 \quad\Longleftrightarrow\quad
 i\equiv j\pmod{N_g}
 \qquad(i,j\ge0).
 \label{eq:exact-projective-period}
\end{equation}
In particular,
\begin{equation}
 r\le N_g\le\#\mathbf P^{r-1}(\F_q)
 =\frac{q^r-1}{q-1}.
 \label{eq:Ng-range}
\end{equation}
There is a unique $c_g\in\F_q^*$ such that
\begin{equation}
 T_g^{N_g}=c_gI_r,\qquad
 Q_{N_g}=c_gQ_0,\qquad
 g(x)\mid x^{N_g}-c_g.
 \label{eq:Ng-scalar-relation}
\end{equation}
Equivalently,
\begin{equation}
 \boxed{\ 
 \begin{aligned}
 N_g
 &=\min\{N\ge1:\exists\,c\in\F_q^*,\ g(x)\mid x^N-c\}\\
 &=\min\{\deg f:0\ne f\in\F_q[x],\ g\mid f,\ \wt(f)=2\}.
 \end{aligned}
 \ }
 \label{eq:Ng-first-binomial}
\end{equation}
Thus $N_g$ is the degree of the first binomial multiple of $g$.
Consequently,
\begin{equation}
 \boxed{\ \Schur{(N_g-r+1)}(g)=0\ },
 \label{eq:Ng-Schur-zero}
\end{equation}
and
\begin{equation}
 \tau(g)\le N_g-r+1,\qquad
 \sigma_r(g)\le N_g\le\frac{q^r-1}{q-1}.
 \label{eq:Ng-threshold-bounds}
\end{equation}
Thus $\mathcal C_g(n)$ can be MDS only for $n\le N_g$.
\end{proposition}

\begin{proof}
The vector $Q_0$ is cyclic for $T_g$, since
$Q_0,Q_1,\ldots,Q_{r-1}$ is the power basis.  If
$[Q_i]=[Q_j]$ with $i\ge j$, then
$T_g^iQ_0=cT_g^jQ_0$ for some $c\in\F_q^*$.  Since $T_g$ is invertible,
applying $T_g^{-j}$ gives $T_g^{i-j}Q_0=cQ_0$.  Commuting
$T_g^{i-j}$ with the powers of $T_g$ shows that it acts by $c$ on the
entire power basis; hence $T_g^{i-j}=cI_r$.  The converse is immediate, and
\eqref{eq:exact-projective-period} follows.  The first $r$ orbit points are
distinct, while the orbit lies in $\mathbf P^{r-1}(\F_q)$, proving
\eqref{eq:Ng-range}.

By the definition of $N_g$, there is a unique $c_g\in\F_q^*$ with
$T_g^{N_g}=c_gI_r$.  Since the minimal polynomial of the cyclic companion
matrix $T_g$ is $g$, this identity is equivalent to
$g\mid x^{N_g}-c_g$, and it also gives $Q_{N_g}=c_gQ_0$.  Therefore
\[
 g\mid x^N-c
 \quad\Longleftrightarrow\quad
 T_g^N=cI_r
 \qquad(N\ge1,\ c\in\F_q^*).
\]
It follows that any such $N$ is divisible by $N_g$, proving the first
minimum in~\eqref{eq:Ng-first-binomial}.  More generally, if
$f=ax^u+bx^v$ with $a,b\ne0$ and $u>v\ge0$ is divisible by $g$, then
\[
 aT_g^u+bT_g^v=0.
\]
Since $T_g$ is invertible, $T_g^{u-v}=-(b/a)I_r$.  Hence
$N_g\mid(u-v)$ and $\deg f=u\ge N_g$; the binomial
$x^{N_g}-c_g$ gives equality.  This proves the second minimum in
\eqref{eq:Ng-first-binomial}.

Finally,
\[
 \det(Q_0,Q_1,\ldots,Q_{r-2},Q_{N_g})=0,
\]
where the middle list is empty for $r=2$.  The reduced index set
$\{1,\ldots,r-2,N_g\}$ corresponds to the one-row partition
$(N_g-r+1)$.  Theorem~\ref{thm:reduced-Schur} proves
\eqref{eq:Ng-Schur-zero}; the threshold bounds follow from the definitions
and Theorem~\ref{thm:sigma-tau}.  The vanishing factor lies in the outer
 Schur layer introduced at length $N_g+1$.
\end{proof}

The projective-space bound in~\eqref{eq:Ng-range} is elementary.  The
factorwise certificate~\eqref{eq:Ng-Schur-zero}, which locates the first
projective repetition in the specific outer Schur layer introduced at
length $N_g+1$, is the finer information supplied by the rectangular Schur
system.  Thus~\eqref{eq:Ng-threshold-bounds} compares the degree of the first
multiple having at most $r$ terms with the degree of the first binomial
multiple; Example~\ref{ex:irreducible-cubic} shows that the comparison can
be strict.

\begin{proposition}[Root-ratio and multiplicity formula]
\label{prop:root-ratio-projective-period}
Let $g\in\F_q[x]$ be monic of degree $r\ge2$ with $g(0)\ne0$, let $N_g$ be
its companion projective period, and put $p=\operatorname{char}\F_q$.  Write
the factorization of $g$ over $\overline{\F}_q$ as
\begin{equation}
 g(x)=\prod_{\nu=1}^s(x-\lambda_\nu)^{m_\nu},
 \qquad \lambda_\nu\ne0,
 \label{eq:root-multiplicity-factorization-period}
\end{equation}
where the $\lambda_\nu$ are pairwise distinct.  Put
\begin{equation}
 m_*:=\max_{1\le\nu\le s}m_\nu,
 \qquad
 P_g:=\min\{p^a:a\ge0,\ p^a\ge m_*\},
 \label{eq:Pg-definition}
\end{equation}
and
\begin{equation}
 R_g:=\operatorname*{lcm}_{1\le\nu\le s}
       \ord\!\left(\frac{\lambda_\nu}{\lambda_1}\right).
 \label{eq:Rg-definition}
\end{equation}
Here all multiplicative orders are taken in $\overline{\F}_q^*$; when
$s=1$, we set $R_g=1$.
Then $R_g$ is independent of the choice of the reference root $\lambda_1$,
$p\nmid R_g$, and
\begin{equation}
 \boxed{\ N_g=P_gR_g\ }.
 \label{eq:Ng-root-multiplicity-formula}
\end{equation}
Moreover, the scalar in~\eqref{eq:Ng-scalar-relation} is
\begin{equation}
 c_g=\lambda_1^{N_g}=\cdots=\lambda_s^{N_g}\in\F_q^*.
 \label{eq:cg-root-formula}
\end{equation}
Consequently,
\begin{equation}
 \sigma_r(g)\le P_gR_g\le\frac{q^r-1}{q-1}.
 \label{eq:PgRg-threshold-bound}
\end{equation}
\end{proposition}

\begin{proof}
Let $H_g$ be the finite subgroup of $\overline{\F}_q^*$ generated by all pairwise
root ratios $\lambda_\nu/\lambda_\mu$.  For any choice of reference root,
the ratios to that root generate $H_g$, so $R_g$ is the exponent of $H_g$
and is independent of the choice.  Every finite-order element of
$\overline{\F}_q^*$ has order prime to $p$, whence $p\nmid R_g$.

Frobenius permutes the distinct roots of $g$.  If
$\lambda_1^q=\lambda_j$, then, with $d:=\lambda_1^{R_g}$,
\[
 d^q=(\lambda_1^q)^{R_g}=\lambda_j^{R_g}
 =\lambda_1^{R_g}=d.
\]
Thus $d\in\F_q^*$.  Every $\lambda_\nu$ is a root of $x^{R_g}-d$, and this
polynomial is squarefree because $p\nmid R_g$.  Since $P_g$ is a power of
$p$ and $P_g\ge m_\nu$ for every $\nu$,
\[
 g(x)\mid (x^{R_g}-d)^{P_g}
 =x^{R_gP_g}-d^{P_g}.
\]
The first minimum in~\eqref{eq:Ng-first-binomial} therefore gives
$N_g\le R_gP_g$.

Conversely, suppose that $g(x)\mid x^N-c$ with $c\in\F_q^*$, and write
$N=p^bM$ with $p\nmid M$.  Over $\overline{\F}_q$ one has
\[
 x^N-c=(x^M-c_0)^{p^b}
\]
for the unique $p^b$-th root $c_0$ of $c$.  Since $\F_q$ is perfect,
$c_0\in\F_q^*$.  The polynomial $x^M-c_0$ is squarefree, so each root of
$x^N-c$ has multiplicity exactly $p^b$.
Consequently $p^b\ge m_*$ and hence $P_g\mid p^b$.  Also
$(\lambda_\nu/\lambda_1)^N=1$ for every $\nu$.  The orders of these ratios
are prime to $p$, so they divide $M$, and therefore $R_g\mid M$.  It follows
that $P_gR_g\mid N$.  Since $N_g$ is the least exponent for which such a
binomial multiple exists,~\eqref{eq:Ng-root-multiplicity-formula} follows.
Finally, evaluating $x^{N_g}-c_g$ at every distinct root gives
$c_g=\lambda_1^{N_g}=\cdots=\lambda_s^{N_g}$; this common value lies in
$\F_q^*$ also by the construction above.
\end{proof}

In particular, if $g$ is squarefree, then $N_g=R_g$.  If
$g(x)=(x-\beta)^r$ with $\beta\in\F_q^*$, then
\begin{equation}
 N_g=p^{\lceil\log_p r\rceil}.
 \label{eq:pure-power-projective-period}
\end{equation}

\begin{remark}[The data controlling the projective period]
\label{rem:projective-period-data}
The multiplicity pattern enters $N_g$ only through its largest part $m_*$,
whereas the semisimple contribution is the exponent of the projective
root-ratio group.  Thus polynomials with the same projective root
configuration and the same maximal root multiplicity have the same
companion projective period, even when their remaining multiplicities are
distributed differently.  The formula also gives, for $c\in\F_q^*$,
\begin{equation}
 N_{g_c}=N_g,
 \qquad
 N_{g^\#}=N_g,
 \label{eq:Ng-symmetry}
\end{equation}
because dilation preserves root ratios, while reciprocity inverts them, and
both operations preserve root multiplicities.  This agrees with the
$\Gm\rtimes C_2$-symmetry of the first-failure filtration.
\end{remark}

\begin{corollary}[Projective period on the irreducible stratum]
\label{cor:irreducible-projective-order}
Suppose in addition that $g$ is irreducible, let
$\alpha\in\F_{q^r}$ be a root, and put
\begin{equation}
 N_\alpha:=\min\{N>0:\alpha^N\in\F_q^*\}.
 \label{eq:Nalpha}
\end{equation}
Then
\begin{equation}
 N_g=N_\alpha
 =\frac{m}{\gcd(m,q-1)},
 \qquad
 m:=\ord_{\F_{q^r}^*}(\alpha).
 \label{eq:irreducible-projective-period}
\end{equation}
In particular, all conclusions of
Proposition~\ref{prop:general-projective-order} hold with $N_g=N_\alpha$.
\end{corollary}

\begin{proof}
The polynomial $g$ is squarefree, and its roots are
$\alpha,\alpha^q,\ldots,\alpha^{q^{r-1}}$.  Relative to $\alpha$, every root
ratio is a power of $\alpha^{q-1}$, and the ratio $\alpha^{q-1}$ itself
occurs.  Proposition~\ref{prop:root-ratio-projective-period} therefore gives
\[
 N_g=R_g=\ord(\alpha^{q-1})
 =\frac{m}{\gcd(m,q-1)}=N_\alpha,
\]
which is~\eqref{eq:irreducible-projective-period}.
\end{proof}

\begin{remark}[Finite first-failure stratification over finite fields]
\label{rem:finite-first-failure-stratification}
For $r\ge2$, put
\[
 \Pi_{q,r}:=\#\PP^{r-1}(\F_q)=\frac{q^r-1}{q-1}.
\]
Corollary~\ref{cor:first-failure-strata} and
Proposition~\ref{prop:general-projective-order} give the finite point-set
stratification
\begin{equation}
 D(A_0)(\F_q)
 =\bigsqcup_{N=r}^{\Pi_{q,r}}\mathcal X_{N,r}(\F_q).
 \label{eq:finite-first-failure-stratification}
\end{equation}
In particular, the infinite-threshold stratum has no $\F_q$-points.
\end{remark}

\section{Specializations and examples}
\label{sec:examples}

The examples below are organized by the rectangular Schur system.  They
recover two boundary formulas, display the first nontrivial rectangle, and
show how an irreducible example simultaneously records a Frobenius alphabet
and a sparse threshold.

\begin{example}[The first boundary $n=r+1$]
\label{ex:first-boundary}
The rectangle is $1^{r-1}$.  Its nonempty partitions are the columns
$(1^j)$, and $s_{(1^j)}=e_j$.  Hence
\begin{equation}
 \mathscr D_{r+1,r}(g)=a_0\prod_{j=1}^{r-1}e_j(\Lambda_g)
 =\pm\prod_{j=0}^{r-1}a_j.
 \label{eq:first-boundary}
\end{equation}
Thus $\mathcal C_g(r+1)$ is MDS exactly when every coefficient of $g$ is
nonzero.
\end{example}

\begin{example}[The second boundary $n=r+2$]
\label{ex:second-boundary}
 Every partition in $2^{r-1}$ has a unique expression $(2^i,1^j)$ with
 $i,j\ge0$ and $i+j\le r-1$, and dual Jacobi--Trudi gives
\begin{equation}
 s_{(2^i,1^j)}=e_{i+j}e_i-e_{i+j+1}e_{i-1}.
 \label{eq:hook-boundary-Schur}
\end{equation}
The cases $i=0$ give the column factors $e_j$.  For $i\ge1$, the change of
indices
\[
 u=r-i-j,\qquad v=r-i+1
\]
is a bijection with $1\le u<v\le r$, and, up to sign, the corresponding
factor is $a_u a_{v-1}-a_{u-1}a_v$.  Thus, with $a_r=1$, the rectangular
product becomes
\begin{equation}
 \mathscr D_{r+2,r}(g)
 =\pm\left(\prod_{j=0}^{r-1}a_j\right)
 \prod_{1\le i<j\le r}
 (a_i a_{j-1}-a_{i-1}a_j).
 \label{eq:second-boundary}
\end{equation}
Therefore the code is MDS precisely when $a_0,\ldots,a_{r-1}$ are nonzero
and
\[
 \frac{a_1}{a_0},\frac{a_2}{a_1},\ldots,
 \frac{a_r}{a_{r-1}}
\]
are pairwise distinct.  The Schur system thus recovers from one rectangle
construction the two boundary criteria obtained in
\cite[Theorem~A.1\textup{(iii)--(iv)}]{LiYuanCyclicOrbits}.
\end{example}

\begin{example}[The cubic length-six rectangle]
\label{ex:cubic-six-product}
Let $r=3$, $n=6$, and $g(x)=x^3+ax^2+bx+c$.  The rectangle is $3^2$, and
its canonical product is
\begin{align}
 \mathscr D_{6,3}(g)={}&abc(a^2-b)(a^3-2ab+c)(ab-c)
 \notag\\
 &\cdot(a^2b-ac-b^2)(ac-b^2)
 \notag\\
 &\cdot(a^2c-ab^2+bc)(b^3-2abc+c^2).
 \label{eq:cubic-six-product}
\end{align}
The displayed determinantal polynomial is exactly the Schur product for
$3^2$.
\end{example}

\begin{example}[The repeated-root type $(2,1)$]
\label{ex:two-one}
For the alphabet $(\alpha,\alpha,\beta)$,
\begin{equation}
 s_{(2,1)}(\alpha,\alpha,\beta)
 =2\alpha(\alpha+\beta)^2.
 \label{eq:local-two-one}
\end{equation}
This local characteristic-$2$ collapse is the first case of the exact
classification
\[
 \mathcal B_{n,(2,1)}=\{p:2p\le n-1\}.
\]
For example, the bad set is $\{2\}$ at $n=5,6$ and $\{2,3\}$ at
$n=7,8$.
\end{example}

\begin{example}[An irreducible Frobenius--Schur example]
\label{ex:irreducible-cubic}
This is a complete worked example of Theorem~\ref{thm:sigma-tau}.
Over $\F_{11}$ let
\begin{equation}
 g(x)=x^3+x^2+2x+10.
 \label{eq:irreducible-example-g}
\end{equation}
The cubic has no root in $\F_{11}$ and is irreducible.  If $\alpha$ is one
of its roots, the Frobenius alphabet is
$(\alpha,\alpha^{11},\alpha^{121})$.  Evaluating the rectangular product
\eqref{eq:cubic-six-product}, the ten Schur factors indexed by
\[
 \varnothing,(1),(2),(3),(1,1),(2,1),(3,1),(2,2),(3,2),(3,3)
\]
take the respective values
\[
 1,10,10,4,2,8,10,5,4,2.
\]
Together with $a_0=10$, their product gives
\begin{equation}
 \mathscr D_{6,3}(g)=8\ne0\quad\text{in }\F_{11},
 \label{eq:irreducible-example-D}
\end{equation}
so $\mathcal C_g(6)$ is MDS and,
by~\cite[Corollary~5.5]{LiYuanCyclicOrbits}, automatically non-GRS.

At length seven, the new outer-layer factors
$\Schur{(4)},\Schur{(4,1)},\ldots,\Schur{(4,4)}$ have values
\[
 8,9,8,4,0,
\]
so the first failure occurs there.  The remainder columns satisfy
\[
 Q_5=(10,1,4)^T,
 \qquad
 Q_6=(4,2,8)^T=2Q_5+6Q_0.
\]
Equivalently,
\begin{equation}
 g(x)\mid x^6-2x^5-6
 =(x^3+8x^2+x+6)g(x)
 \quad\text{in }\F_{11}[x].
 \label{eq:irreducible-example-sparse}
\end{equation}
Thus $\sigma_3(g)=6$.  The dependent set
$\{0,5,6\}$ corresponds to $\kappa=(4,4)$, so
$\Schur{(4,4)}(g)=0$, $\tau(g)=4$, and
$\sigma_3(g)=3+4-1$ exactly as in~\eqref{eq:sigma-tau}.  Repeated reduction
using $x^3\equiv10x^2+9x+1\pmod g$ gives
\[
 x^9\equiv5+6x+4x^2,\qquad
 x^{10}\equiv4+8x+2x^2,\qquad
 x^{19}\equiv1\pmod g.
\]
Thus $\alpha^{19}=1$.  Since $19$ is prime and irreducibility gives
$\alpha\notin\F_{11}$, the multiplicative order of $\alpha$ is $19$.
Moreover $\gcd(19,|\F_{11}^*|)=\gcd(19,10)=1$, so no smaller positive power
of $\alpha$ is scalar over $\F_{11}$.  Hence its projective order is
$N_\alpha=19$, illustrating that
$\sigma_3(g)\le N_\alpha$ can be strict.
\end{example}

\section{Conclusion}
\label{sec:conclusion}

The cyclic-pair quotient identifies coefficient space with the companion
slice, whose remainder orbit is a smooth complete intersection in the
Grassmannian big cell.  Its normalized maximal Pl\"ucker coordinates form one
rectangular Schur system; hence the coefficient-space MDS locus of this
principal-ideal family is the pullback of the uniform-matroid open set, with
a flat principal Cartier boundary of explicit degree and flat outer layers.
The full Pl\"ucker boundary has the same support as the canonical Schur
boundary, while its Cartier multiplicities are explicitly determined and
generally larger.  Dilation and reciprocity give the two basic symmetries.

The same rectangles control multiplicity and Frobenius strata and identify
first failure with the first sparse multiple through
$\sigma_r(g)=r+\tau(g)-1$.  The first-failure Cartier layers are flat and,
for $N>r$, split off scaling.  Every structural bad characteristic on a
root-multiplicity stratum is smaller than the code length.  The GRS
applications combine these
Schur--Pl\"ucker results with the pure-power criterion and the
rational-normal-curve classification and counting results
in~\cite{LiYuanCyclicOrbits}.  Over an algebraically closed field, for
$r\ge3$ and
$N\ge r+3$, the closure of the GRS-type coefficient set in every nonempty
first-failure layer has dimension at most one, so the non-GRS locus contains
a dense open subset of every irreducible component.  This separates the
coefficient-space MDS locus of this principal-ideal family from the
rational-normal-curve geometry governing its GRS-type subset.

\appendix

\section{Quadratic Frobenius reconstruction}
\label{app:quadratic-reconstruction}

We give the coefficient reconstructions used in
Propositions~\ref{prop:all-quadratic-coeff} and
\ref{prop:mixed-quadratic-coeff}; their MDS hypotheses are supplied by the
rectangular Schur system.

\subsection{The all-quadratic stratum}
\label{app:all-quadratic-reconstruction}

The nonsplit family \textup{(N0)}
in~\cite[Theorem~5.3, family~\textup{(N0)}]{LiYuanCyclicOrbits} has roots
\[
 \gamma,\gamma t,\ldots,\gamma t^{2v-1},
 \qquad
 t^{q+1}=1,\quad \ord(t)\ge n,\quad
 \gamma^q=\gamma t^{2v-1}.
\]
Frobenius pairs the exponents $j$ and $2v-1-j$.  Thus, after reordering
the irreducible quadratic factors, their traces and norms are
\[
 \tau_j=\gamma t^j+\gamma^qt^{-j},
 \qquad P_j=\gamma^{q+1}
 \qquad(0\le j\le v-1).
\]
Put $\sigma=t+t^{-1}\in\F_q$ and $P=\gamma^{q+1}$.  Then
$T^2-\sigma T+1$ is the irreducible polynomial of $t$, all $P_j=P$, and
\[
 \tau_{j+2}=\sigma\tau_{j+1}-\tau_j.
\]
If the recurrence is extended by
$\tau_v=\sigma\tau_{v-1}-\tau_{v-2}$, the Frobenius relation
$\gamma^q=\gamma t^{2v-1}$ is equivalent to $\tau_v=\tau_{v-1}$, or
$\tau_{v-2}=(\sigma-1)\tau_{v-1}$.  Direct expansion gives
\[
 \tau_0^2+\tau_1^2-\sigma\tau_0\tau_1=(4-\sigma^2)P.
\]
These identities prove the necessity of conditions
\textup{(1)}--\textup{(5)} in
Proposition~\ref{prop:all-quadratic-coeff}.

Conversely, assume those five conditions and let $t$ be a root of
$T^2-\sigma T+1$.  Irreducibility gives $t^q=t^{-1}$ and
$t\ne t^{-1}$, also in characteristic $2$, so put
\[
 \gamma=\frac{\tau_1-\tau_0t^{-1}}{t-t^{-1}}.
\]
Since
\[
 4-\sigma^2=-(t-t^{-1})^2\ne0,
\]
a direct calculation from the definition of $\gamma$ and
$t^q=t^{-1}$ gives
\[
 \tau_0^2+\tau_1^2-\sigma\tau_0\tau_1
 =(4-\sigma^2)\gamma^{q+1}.
\]
Condition \textup{(5)} therefore gives $\gamma^{q+1}=P$, and the recurrence
yields
\[
 \tau_j=\gamma t^j+\gamma^qt^{-j}
 \qquad(0\le j\le v).
\]
The terminal condition $\tau_v=\tau_{v-1}$ then implies
$\gamma^q=\gamma t^{2v-1}$.  Hence the reconstructed roots form the
displayed orbit, and the classification
in~\cite[Theorem~5.3, family~\textup{(N0)}]{LiYuanCyclicOrbits} proves the
claimed GRS equivalence.

\subsection{The one-linear-plus-quadratic stratum}
\label{app:mixed-quadratic-reconstruction}

Let $t$ be a norm-one element of order at least $n$, and put
$\sigma=t+t^{-1}$.  The recurrence~\eqref{eq:Dickson-recurrence} gives
\[
 \mathcal D_j(\sigma)=t^j+t^{-j}.
\]
The centered nonsplit orbit \textup{(N1)}
\[
 \{\beta t^{-v},\ldots,\beta t^{-1},\beta,
   \beta t,\ldots,\beta t^v\}
\]
therefore has factorization
\[
 (x-\beta)\prod_{j=1}^{v}
 \bigl(x^2-\beta\mathcal D_j(\sigma)x+\beta^2\bigr).
\]
This proves the necessity of~\eqref{eq:mixed-coefficient-test}.
Conversely, irreducibility of $T^2-\sigma T+1$ gives $t^q=t^{-1}$ and
$t\ne t^{-1}$, including in characteristic $2$.  Equations
\eqref{eq:Dickson-recurrence} and~\eqref{eq:mixed-coefficient-test}
recover exactly this orbit.  The order condition makes its points distinct,
and the classification
in~\cite[Theorem~5.3, family~\textup{(N1)}]{LiYuanCyclicOrbits} proves
sufficiency.

\end{document}